\documentclass[a4paper, USenglish]{lipicsoid-2019}

\hideLIPIcs
\usepackage{etex}
\usepackage{amsmath}
\usepackage{amssymb,amsthm,url}
\usepackage{stmaryrd}
\usepackage{todonotes}
\usepackage{comment}
\usepackage{todonotes}

\usepackage[numbers, sort&compress]{natbib} 
\usepackage{doi} 
\usepackage{bibentry}
\nobibliography*

\newtheorem{open}{Open question}

\hypersetup{
  colorlinks=true,
  linkcolor=black,
  citecolor=black,
  filecolor=black,
  urlcolor=[rgb]{0,0.1,0.5},
  pdftitle={Error-Sensitive Proof-Labeling Schemes},
  pdfauthor={}
}

\newcommand{\cL}{{\cal L}}
\newcommand{\cP}{{\cal P}}
\newcommand{\cO}{{\cal O}}
\newcommand{\dist}{\mbox{\rm dist}}
\newcommand{\ID}{\mbox{\sf ID}}
\newcommand{\prover}{\mathbf{p}}
\newcommand{\verif}{\mathbf{v}}
\newcommand{\accept}{\mbox{\sl accept}}
\newcommand{\reject}{\mbox{\sl reject}}
\newcommand{\MST}{\mbox{\sc mst}}
\newcommand{\ST}{\mbox{\sc ST}}
\newcommand{\LOCAL}{\mbox{\sc local}}

\title{Error-Sensitive Proof-Labeling Schemes\footnote{A preliminary version of this paper appeared in the proceedings of the 31st International Symposium on Distributed Computing (DISC), October 16-20, 2017, Vienna, Austria.}}

\author{Laurent Feuilloley}%
{Univ. Lyon, Université Lyon 1, LIRIS UMR CNRS 5205, F-69621, Lyon, France}%
{laurent.feuilloley@univ-lyon1.fr}%
{http://orcid.org/0000-0002-3994-0898}%
{}

\author{Pierre Fraigniaud}%
{Institut de Recherche en Informatique Fondamentale (IRIF), Universit\'e de Paris and CNRS, France}%
{pierre.fraigniaud@irif.fr}%
{}%
{}

\authorrunning{L. Feuilloley and P. Fraigniaud}
\funding{ANR ESTATE (ANR-16-CE25-0009-03), and ANR GrR (ANR-18-CE40-0032).}



\begin{document}

\maketitle

\begin{abstract}

Proof-labeling schemes are known mechanisms providing nodes of networks with \emph{certificates} that can be \emph{verified} locally by distributed algorithms. Given a boolean predicate on network states, such schemes enable to check whether the predicate is satisfied by the actual state of the network, by having nodes interacting  with their neighbors only. Proof-labeling schemes are typically designed for enforcing fault-tolerance, by making sure that if the current state of the network is illegal with respect to some given predicate, then at least one node will detect it. Such a node can raise an alarm, or launch a recovery procedure enabling the system to return to a legal state. 

In this paper, we introduce \emph{error-sensitive} proof-labeling schemes. These are proof-labeling schemes which guarantee that the number of nodes detecting illegal states is linearly proportional to the Hamming distance between the current state and the set of legal states. By using  error-sensitive proof-labeling schemes, states which are far from satisfying the predicate will be detected by many nodes. We provide a structural characterization of the set of boolean predicates on network states for which there exist error-sensitive  proof-labeling schemes. This characterization allows us to show that classical predicates such as, e.g., cycle-freeness, and leader admit error-sensitive  proof-labeling schemes, while others like regular subgraphs do not. We also focus on \emph{compact} error-sensitive  proof-labeling schemes. 
In particular, we show that the known proof-labeling schemes for  spanning tree and minimum spanning tree, using certificates on $O(\log n)$ bits, and on $O\left(\log^2n\right)$ bits, respectively, are error-sensitive, as long as the trees are locally represented by adjacency lists, and not just by parent pointers.  

\end{abstract}

\thispagestyle{empty}
\newpage

\section{Introduction}

In the context of fault-tolerant distributed computing, it is desirable that the computing entities in the system be able to detect whether the system is in a legal state (w.r.t.~some boolean predicate, potentially expressed in various forms of logics) or not.  In the framework of distributed network computing, several mechanisms have been proposed to ensure such a detection (see, e.g., \cite{AfekD02,AfekKY97,AwerbuchPV91,%
BeauquierDDT07,KormanKP10}). Among them, \emph{proof-labeling schemes}~\cite{KormanKP10} are mechanisms enabling failure detection based on additional information provided to the nodes. More specifically, a proof-labeling scheme is composed of a \emph{prover}, and a \emph{verifier}. A prover is a non-trustable oracle that assigns a \emph{certificate} to each node of any given network, and a verifier is a distributed algorithm that locally checks  whether the collection of certificates is a \emph{distributed proof} that the network is in a legal state with respect to a given predicate -- by ``locally'', we essentially mean: by having each node interacting once with its neighbors. 

The prover is actually an abstraction. In practice, the certificates are provided by a distributed algorithm solving some task (see, e.g., \cite{AwerbuchKMPV93,BlinF15,KormanKP10}). For instance, let us consider spanning tree construction, where every node must compute a pointer to a neighboring node such that the collection of pointers form a tree spanning all nodes in the network. In that case, the algorithm in charge of constructing a spanning tree is also in charge of constructing the certificates providing a distributed proof allowing a verifier to check that proof locally. That is, the verifier must either accept or reject at every node, under the following constraints. If the constructed set of pointers form a spanning tree, then the constructed certificates must lead the verifier to accept at every node. Instead, if the constructed set of pointers does not form a spanning tree, then, for every possible certificate assignment to the nodes, at least one node must reject. The rejecting node may then raise an alarm, or launch a recovery procedure. Abstracting the construction of the certificates thanks to a prover enables to avoid delving into the implementation details relative to the distributed construction of the certificates, for focusing attention on whether such certificates exist, and on what should be their forms. The reader is referred to~\cite{BlinFP14} for more details about the connections between proof-labeling schemes and fault-tolerant computing. 

One weakness of proof-labeling schemes is that they may not allow the system running the verifier to distinguish between a global state which is slightly erroneous, and a global state which is completely bogus. In both cases, it is only required that at least one node detects the illegality of the state. In the latter case though, having only one node raising an alarm, or launching a recovery procedure for bringing  the whole system back to a legal state, might be quite inefficient. Instead, if many nodes would detect the errors, then bringing back the system into a legal state may be achieved by a collection of local resets running in parallel, instead of a single reset traversing the whole network sequentially. 

In this paper, we aim at designing \emph{error-sensitive} proof-labeling schemes, which guarantee that system states that are far from being correct can be detected by many nodes. 
More specifically, the distance between two global states of a distributed system is defined as the \emph{Hamming distance} between these two states, i.e., the minimum number of individual states that must be modified in order to move from one global state to the other. 
A proof-labeling scheme is \emph{error-sensitive} if there exists a constant $\alpha>0$ such that, for any erroneous system state $S$, the number of nodes detecting the error is at least $\alpha \, d(S)$, where $d(S)$ is the shortest Hamming distance between $S$ and a correct system state. 
The choice of a linear dependency between the number of nodes detecting the error, and the Hamming distance to legal states is not arbitrary, but motivated by the following two observations. 

\begin{itemize}
\item On the one hand, a linear dependency is somewhat the best that we may hope for in general. 
Indeed, let us consider a $k$-node network $G$ in some illegal state $S$ with $d(S)=d>0$, for which $f(d)$ nodes are detecting the illegality of~$S$, for some function~$f$. Think about vertex-coloring, in which one needs to modify the colors of at least $d$ nodes in order to get a proper coloring.  Then, let us make $n$ copies of $G$ and of its state $S$, potentially linked by $n-1$ additional edges if one insists on connectivity. In the resulting $kn$-node network~$G'$, at most $O(n\cdot f(d))$ nodes are detecting the non legality of the global state~$S'$ of~$G'$. However, $S'$ is typically at distance $\Omega(n\cdot d)$ from any legal state (think again about proper vertex-coloring). It follows that, essentially, $f(nd)\leq n \cdot f(d)$, that is, the number of nodes detecting an error cannot grow faster than linearly with the distance to the legal states. 

\item On the other hand, while a sub-linear dependency may still be useful in some contexts, this would be insufficient in others. 
For instance, let us consider the same construction as above, with $f(d)=d^\alpha$ for some $\alpha<1$. As $n$ grows to infinity, the ratio between the  number of nodes $f(nd)=(nd)^\alpha$  that are asked to detect errors in~$S'$ and the number of nodes $nk$ in the network~$G'$ goes to zero. This results in significantly decreasing the impact of having more than one node detecting the illegality of the current system state, as the number of nodes detecting errors becomes negligible anyway in front of the total number of nodes, even for scenarios in which the distance to legal states grows linearly with the total number of nodes. 
\end{itemize}

\subsection{Our results}

We consider boolean predicates on graphs with labeled nodes, as in, e.g., \cite{NaorS95}. Given a graph~$G$, a labeling of $G$ is a function $\ell:V(G)\to \{0,1\}^*$ assigning binary strings to nodes. A \emph{labeled graph} is a pair $(G,\ell)$ where $G$ is a graph, and $\ell$ is a labeling of $G$. Given a boolean predicate $\cP$ on labeled graphs, the \emph{distributed language} associated to $\cP$ is: 
\[
\cL=\{(G,\ell) \; \mbox{satisfying} \; \cP\}. 
\]
It is known that every (Turing decidable) distributed language admits a proof-labeling scheme~\cite{GoosS16,KormanKP10}. We show that the situation is radically different when one is interested in error-sensitive proof-labeling schemes. 
In particular, not all distributed languages admit an error-sensitive proof-labeling scheme. 
Moreover, the existence of error-sensitive proof-labeling schemes for the solution of a distributed task is very much impacted by the way the solution is encoded. 
For instance, in the case of spanning tree construction, we show that asking every node to produce a single pointer to its parent in the tree  cannot be certified in an error-sensitive manner, while asking every node to produce the list of its neighbors in the tree can  be certified in an error-sensitive manner. 

Our first main result is a structural characterization of the distributed languages for which there exist error-sensitive  proof-labeling schemes. 
Namely, a distributed language admits an error-sensitive proof-labeling scheme if and only if it is \emph{locally stable}. The notion of local stability is purely structural. 
Roughly, a distributed language $\cL$ is locally stable if a labeling $\ell$ resulting from copy-pasting parts of correct labelings to different subsets $S_1,\dots,S_k$ of nodes in a graph $G$ results in a labeled graph $(G,\ell)$  that is not too far from being legal. 
Here ``not too far'' means that the Hamming distance between $(G,\ell)$ and $\cL$ is proportional to the size of the boundary of the subsets $S_1,\dots,S_k$ in $G$, and not to the size of these subsets.
For the sake of concreteness, let us give an intuition about why 
a spanning tree encoded by a list of neighbors is a locally stable language. 
Consider a graph partitioned into $k$ connected induced subgraphs such that only a small fraction of the nodes are on the boundary of a subgraph (i.e., are having a neighboring node in another subgraph).
Now, let us consider a spanning tree in each of the $k$ subgraphs.  
The union of these spanning trees is not a spanning tree, but it is not far from being a spanning tree. 
Indeed, it is acyclic, and we can simply add edges to make it connected. To do so, we only modify the adjacency list of vertices that are on the boundaries, thus the distance between the original instance and the modified one is smaller than the sum of the sizes of the boundaries. (This example is actually simplified, as it assumes that the trees in each component are correct, i.e., connected, and without cycles, which may not be the case.)
Our characterization allows us to show that important distributed languages (e.g., acyclicity, leader, etc.)  admit error-sensitive  proof-labeling schemes, while some very basic distributed languages (e.g., regular subgraph, etc.) do not admit  error-sensitive  proof-labeling schemes. 

Unfortunately, the error-sensitive schemes constructed for locally stable languages in the proof of our characterization result are not efficient in terms of certificate size. 
We investigate the question of whether it is possible to get error-sensitivity with small certificates.
For this purpose, we focus on two essential languages: spanning tree, which is a building block for many proof-labeling schemes, and minimum spanning tree, which is arguably one of the most important problems in distributed network computing.

We show that the known space-optimal proof-labeling schemes for spanning tree with $O(\log n)$-bit certificates, and for minimum spanning tree (MST) with $O(\log^2 n)$-bit certificates, are both error-sensitive, whenever the trees are encoded at each node by an adjacency list (and not by a single pointer to the parent). Hence, error-sensitivity comes at no cost for spanning tree and MST. 
Proving this result requires to establish some kind of matching between the erroneously labeled nodes and the rejecting nodes. 
Establishing this matching is difficult because, for  both spanning tree and MST, the rejecting nodes might be located far away from the erroneous nodes. Indeed, the presence of certificates helps local detection of errors, but decorrelates the nodes at which the alarms take place from the nodes at which the errors take place.
For example, in an erroneous spanning tree that is disconnected, it may be the case that only one node is detecting the error, and that this node is far from a place where disconnection can be fixed by adding an edge.
(See Section~\ref{sec:conclusion} for a discussion about \emph{proximity-sensitive} proof-labeling schemes). 
In the case of MST, the space-optimal proof-labeling schemes uses $\Theta(\log n)$ independent layers of certification, and this a challenge for error-sensitivity. 
Indeed, because detection and correction could happen in different places, the following scenario cannot be ruled out directly. It could be the case that: 
(1) every layer of certification is broken, but
(2) only one vertex rejects (because all the $\Theta(\log n)$ parallel verifications reject on the same vertex), and 
(3) to fix the instance, we would need to modify the input of $\Theta(\log n)$ different vertices. 
In short, we could have one vertex rejecting but distance $\Theta(\log n)$, which would prevent error-sensitivity.
Our result demonstrates that this situation cannot appear.

\subsection{Related work}

As mentioned before, one important motivation for our work is fault-tolerant distributed computing, with the help of failure detection mechanisms such as proof-labeling schemes. Proof-labeling schemes were introduced in~\cite{KormanKP10}. A tight bound of $\Theta(\log^2n)$ bits on the size of the certificates for certifying MST was established in \cite{KorKP11,KormanK07}. Several variants of proof-labeling schemes have been investigated in the literature, including verification at  distance greater than one~\cite{GoosS16}, and the design of proofs with identity-oblivious certificates~\cite{FraigniaudKP13}. Connections between proof-labeling schemes and the design of distributed (silent) self-stabilizing algorithms were studied in~\cite{BlinFP14}. Extensions of proof-labeling schemes for the design of (non-silent) self-stabilizing algorithms were investigated in~\cite{KormanKM15}. In all these work, the number of nodes susceptible to detect an incorrect configuration is not considered, and the only constraint imposed on the error-detection mechanism is that an erroneous configuration must be detected by at least one node. Our work requires the number of nodes detecting an erroneous configuration to grow linearly with the number of errors. As mentioned earlier, having several nodes detecting an error allows to launch a reset from several nodes at once. See \cite{BoulinierPV04, DevismesJ19} for references on such collaborative resets. 
Note that taking into account how far from a correct configuration the network is, is not a new idea. Indeed there is a literature on fault-containment or fault-locality (see, e.g., \cite{GhoshGHP07, KuttenP99}), where the focus is on having correction algorithms that use little resources if there are just a few faults, or if these fault are grouped together somehow. 
In particular, \cite{GhoshGHP07} defines a notion of ``small-scale'' faults, for which the system can converge to a correct solution without modifying the states of the nodes that are far from the faulty nodes. 
Our work has a different objective, that is, making sure that incorrect global states resulting from many incorrect local states must be detected by many nodes, while incorrect global states resulting from just a few incorrect local states may be detected by few nodes only.

A line of work closely related to this paper is \emph{property testing}. Centralized property testing for graph properties was investigated in numerous papers (see~\cite{Goldreich11p,Goldreich11q} for an introduction to the topic). Distributed property testing has been introduced in~\cite{BrakerskiP11}, and formalized in~\cite{Censor-HillelFS19} (see also, e.g., ~\cite{EvenFFGLMMOORT17,FraigniaudRST16}). In both centralized and decentralized property testing, the decision regarding whether the labeled input graph satisfies a given property (e.g., cycle-freeness) is typically relaxed: if the graph satisfies the property then all centralized queries, or all nodes must accept, and if the graph is far from satisfying the property (e.g., it contains many cycles), then at least one centralized query, or at least one node must reject. The notion of ``far'' depends on the context. The one adopted in the distributed setting is defined by the sparse model, specifying that a graph is $\epsilon$-far from satisfying a property if any modification up to a fraction~$\epsilon$ of the edges results in a graph that is still not satisfying the property. The goal is then to distinguish  graphs that satisfy the property from graphs that are $\epsilon$-far from satisfying the property. In some sense, property testing can be viewed as efficiently approximating the solution of a hard problem (e.g., NP-hard), while proof-labeling schemes can be viewed as establishing that the problem is complete (e.g., NP-complete). Centralized property testing was actually extended to a non-deterministic setting~\cite{GurR18,LovaszV13} in which the centralized algorithm is provided with a centralized certificate. In error-sensitive proof-labeling schemes, we try to get the best of both worlds, that is, if the input graph is far from satisfying the property, then, whatever are the certificates provided to the nodes by the prover, a large number of nodes must reject the instance. The farness notion used in distributed property testing refers to the edges, while we use a farness notion related to the nodes, but the two notions are essentially the same in bounded-degree graphs. 

From a higher perspective, our approach aims at closing the gap between local distributed computing and centralized computing in networks, by studying distributed error-detection mechanisms that perform locally, but generate individual outputs that are related to the global correctness of the system at hand.  As such, it is worth mentioning other efforts in the same direction, including especially work in the context of centralized local computing, like, e.g.,~\cite{EvenMR18,GoosHLMS16,ParnasR07}. Finally, distributed property testing and proof-labeling schemes are different forms of distributed decision mechanisms (e.g., distributed interactive proofs~\cite{KolOS18,NaorPY20}), which have been investigated under various models for distributed computing. We refer to \cite{FeuilloleyF16} for a survey on distributed decision, and to \cite{Feuilloley21} for a more introductory text.

\section{Model and definitions}

Throughout the paper, all graphs are assumed to be connected and simple (no self-loops, and no parallel edges). Given a node $v$ of a graph $G$, we denote by $N(v)$ the open neighborhood of $v$, i.e., the set of neighbors of $v$ in $G$. In some contexts (as, e.g., MST), the considered graphs may be edge-weighted.  

All results in this paper are stated in the classical \LOCAL\/ model~\cite{Peleg00} for distributed network computing, where networks are modeled by undirected graphs whose nodes model the computing entities, and edges model the communication links. Recall that the  \LOCAL\/ model assumes that nodes are given distinct identities (a.k.a.~IDs), and that computation proceeds in synchronous rounds. All nodes simultaneously start executing the given algorithm. At each round, nodes exchange messages with their neighbors, and perform individual computation. There are no limits placed on the message size, nor on the amount of computation performed at each round. Specifically, we are  interested in \emph{proof-labeling schemes}~\cite{KormanKP10}, which are  well established mechanisms enabling to locally detect inconsistencies in the global states of networks with respect to some given boolean predicate. Such mechanisms involve a verification algorithm which performs in just a single round in the \LOCAL\/ model. In order to recall the definition of proof-labeling schemes, we first recall the definition of \emph{distributed languages}~\cite{FraigniaudKP13}. 

A distributed language is a collection of labeled graphs, that is, a set $\cL$ of pairs $(G,\ell)$ where $G$ is a graph,  and $\ell:V(G)\to\{0,1\}^*$ is a labeling function assigning a binary string to each node of $G$. Such a labeling may encode just a boolean (e.g., whether the node is in a dominating set or not), or an integer (e.g., in graph coloring), or a collection of neighboring IDs (e.g., for locally encoding a subgraph). 
In the latter case, or whenever $\ell$ encodes a set of nodes at each vertex, we may slightly abuse notation by viewing $\ell(v)$ as an actual set of nodes, i.e., by considering $\ell(v)\subseteq V(G)$. 
A distributed language is said to be \emph{constructible} if, for every graph $G$, there exists~$\ell$ such that $(G,\ell)\in\cL$. It is  \emph{Turing decidable} if there exists a (centralized) algorithm which, given $(G,\ell)$ returns whether $(G,\ell)\in \cL$ or not. All distributed languages considered in this paper are always assumed to be constructible and Turing decidable. 

Given a distributed language $\cL$, a proof-labeling scheme for $\cL$ is a prover-verifier pair $(\prover,\verif)$, where $\prover$ is an oracle assigning a certificate function $c:V(G)\to\{0,1\}^*$ to every labeled graph $(G,\ell)\in \cL$, and $\verif$ is a 1-round distributed algorithm\footnote{That is, every node outputs after having communicated with all its neighbors only once.} taking as input at each node $v$ its identity $\ID(v)$, its label $\ell(v)$, and its certificate $c(v)$, such that, for every labeled graph $(G,\ell)$ the following two conditions are satisfied: 

\begin{itemize}
\item If $(G,\ell)\in \cL$ then $\verif$ outputs \accept\/ at every node of $G$ whenever all nodes of $G$ are given the certificates provided by $\prover$;
\item If $(G,\ell)\notin \cL$ then, for every certificate function $c:V(G)\to\{0,1\}^*$, $\verif$ outputs \reject\/ in at least one node of $G$. 
\end{itemize}

The first condition guarantees the existence of certificates allowing the given legally labeled graph $(G,\ell)$ to be globally accepted. The second condition guarantees that the verifier cannot be ``cheated'', that is, an illegally labeled graph will systematically be rejected by at least one node, whatever  ``fake'' certificates are given to the nodes.  It is known that every distributed language has a proof-labeling scheme~\cite{KormanKP10}. 

To define the novel notion of \emph{error-sensitive} proof-labeling schemes, we introduce the following notion of distance between labeled graphs. Let $\ell$ and $\ell'$ be two labelings of a same graph~$G$. 
The \emph{Hamming distance} between $(G,\ell)$ and $(G,\ell')$ is the minimum number of elementary operations required to transform $(G,\ell)$ into $(G,\ell')$, where an elementary operation consists of replacing a node label by another label. That is, the Hamming distance between  $(G,\ell)$ and $(G,\ell')$ is simply
\[
|\{v\in V(G) : \ell(v)\neq \ell'(v)\}|. 
\]
The Hamming distance from a labeled graph $(G,\ell)$ to a language $\cL$ is the minimum, taken over all labelings $\ell'$ of $G$ satisfying $(G,\ell')\in \cL$, of the Hamming distance between $(G,\ell)$ and $(G,\ell')$. 
Note that ``Hamming distance'' is usually defined for words of equal length, by counting the number of characters that must be changed for moving from one word to another word. We use the same terminology in this paper as our distance measures the minimum number of nodes whose states have to be modified to transform a given global state $(G,\ell)$ into another global state $(G,\ell')$. (Instead, distance such as the Edit distance would rather refer to the numbers of edges to be added or deleted for transforming one graph into another.)

Roughly, an error-sensitive proof-labeling scheme satisfies that  the number of nodes that reject a labeled graph $(G,\ell)$ should be (at least) proportional to the distance between $(G,\ell)$ and the considered language. 

\begin{definition}\label{def:error-sensitivity}
A proof-labeling scheme $(\prover,\verif)$ for a language $\cL$ is \emph{error-sensitive} if there exists a constant $\alpha>0$, such that, for every labeled graph $(G,\ell)$, 
\begin{itemize}
\item If $(G,\ell)\in \cL$ then $\verif$ outputs \accept\/ at every node of $G$ whenever all nodes of $G$ are given the certificates provided by $\prover$;\item If $(G,\ell)\notin \cL$ then, for every certificate function $c:V(G)\to\{0,1\}^*$, $\verif$ outputs \reject\/ in at least $\alpha \,\mathsf{d}$ nodes of $G$, where $\mathsf{d}$ is the Hamming distance between $(G,\ell)$ and $\cL$, i.e., $\mathsf{d}=\dist\big((G,\ell),\cL\big)$. 
\end{itemize}
\end{definition}

Note that the nodes rejecting a labeled graph $(G,\ell)$ do not need to be the same for all certificate assignments.
Also note that, as far as this first study of the notion of error-sensitivity is concerned, we are mostly interested in the existence of some constant $\alpha=\Theta(1)$, and not much in the exact value of $\alpha$. However, it is worth keeping in mind that the larger~$\alpha$, the better the error-detection mechanism is, i.e., it is desirable to design protocol for which~$\alpha$ is large. 
For this paper, our focus is a first attempt to explore the notion of error-sensitivity, thus we have not tried  not optimize the constants. Nevertheless, we shall explicitly state what values for the sensitivity~$\alpha$ were used for establishing each of our theorems.

\section{Basic properties of error-sensitive proof-labeling schemes}
\label{sec:warmup}

In this section, we explore basic properties of error-sentivity. First, we show that some  proof-labeling schemes are error-sensitive (Theorem~\ref{prop:acyclic}), but that some other  proof-labeling schemes are not error-sensitive (Theorem~\ref{prop:bad-pls}). More precisely, Theorem~\ref{prop:bad-pls}  shows that even if a language has an error-sensitive proof-labeling scheme, not all proof-labeling schemes for that language have this property.
Second, we show that if a language has an error-sensitive proof-labeling scheme, then the so-called \emph{universal scheme} also has this property (Lemma~\ref{lem:univ}). This implies that for checking whether there exists an error-sensitive scheme for a given language, we can just check whether the universal scheme for that language is error-sensitive. We use this fact for proving that there exist languages that do not have error-sensitive proof-labeling schemes (Theorems~\ref{prop:notallwork} and~\ref{prop:STnot}).

Let us first illustrate the notion of error-sensitive proof-labeling scheme by exemplifying its design for a classic example of distributed languages. Let {\sc acyclic} be the following distributed language (which is a mere relaxation of spanning tree):
\[
\mbox{\sc acyclic} 
	=\Big\{(G,\ell) : \forall v\in V(G), \, \ell(v)\in N(v)\cup \{\bot\}, \; 
	   \mbox{and} \bigcup_{v\in V(G) \, : \,\ell(v)\neq \bot}(v,\ell(v)) \; \mbox{is acyclic} \Big\}
\]
That is, the label of a node is interpreted as a pointer to some neighboring node, or to null. Then $(G,\ell) \in \mbox{\sc acyclic}$ if the subgraph of $G$ described by the set of non-null pointers is acyclic. We show that {\sc acyclic} has an error-sensitive proof-labeling scheme. The proof of this result is easy, as fixing of the labels can be done locally, at the rejecting nodes. Nevertheless, its proof serves as a basic example illustrating the notion of error-sensitive proof-labeling scheme.

\begin{theorem}\label{prop:acyclic}
{\sc acyclic} has an error-sensitive proof-labeling scheme, with sensitivity 1.
\end{theorem}

\begin{proof}
Let $(G,\ell) \in \mbox{\sc acyclic}$. Every node $v\in V(G)$ belongs to an in-tree rooted at a node $r$ such that $\ell(r)=\bot$.  The prover $\prover$ provides every node $v$ with its distance  $d(v)$ to the root of its in-tree (i.e., number of hops to reach the root by following the pointers specified by $\ell$). The verifier $\verif$ proceeds at every node $v$ as follows: first, it checks that $\ell(v)\in N(v)\cup \{\bot\}$; second, it checks that, if $\ell(v)\neq \bot$ then $d(\ell(v))=d(v)-1$, and if $\ell(v)= \bot$ then $d(v)=0$. If all these tests are passed, then $v$ accepts. Otherwise, it rejects. By construction, if $(G,\ell)$ is acyclic, then all nodes accept with these certificates. Conversely, if there is a cycle $C$ in $(G,\ell)$, then let $v$ be a node with maximum value $d(v)$ in $C$. 
Its predecessor in $C$ (i.e., the node $u\in C$ with $\ell(u)=v$) rejects. So $(\prover,\verif)$ is a proof-labeling scheme for \mbox{\sc acyclic}. We show that $(\prover,\verif)$ is error-sensitive. Suppose that $\verif$ rejects $(G,\ell)$ at $k\geq 1$ nodes. Let us replace the label $\ell(v)$ of each rejecting node $v$ by the label $\ell'(v)=\bot$, and keep the labels of all other nodes unchanged, i.e., $\ell'(v)=\ell(v)$ for every node where $\verif$ accepts. We have $(G,\ell')\in \mbox{\sc acyclic}$. Indeed, by construction, the label of every node $u$ in $(G,\ell')$ has a well-formatted label $\ell'(v)\in N(v)\cup\{\bot\}$. Moreover, let us assume, for the purpose of contradiction, that there is a cycle $C$ in $(G,\ell')$. By definition, every node $v$ of this cycle is pointing to  $\ell'(v)\in N(v)$. Thus $\ell'(v)=\ell(v)$ for every node of $C$, from which it follows that no nodes of $C$ was rejecting with $\ell$, a contradiction with the fact that, as observed before, $\verif$ rejects every cycle. Therefore $(G,\ell')\in \mbox{\sc acyclic}$. Hence the Hamming distance between $(G,\ell)$ and {\sc acyclic}  is at most $k$. It follows that $(\prover,\verif)$  is error-sensitive, with sensitivity parameter $\alpha\geq 1$. 
\end{proof}

The definition of error-sensitiveness is based on the existence of a proof-labeling scheme for the considered language. However, two different proof-labeling schemes for the same language may have different sensitivity parameters $\alpha$. In fact, we show that every non-trivial language admits a proof-labeling schemes which is \emph{not} error-sensitive. That is, the following result shows that demonstrating  the existence of a proof-labeling scheme that is \emph{not} error-sensitive for a language does not prevent that language to have another proof-labeling scheme which \emph{is} error-sensitive. We say that a distributed language is \emph{trivially approximable} if there exists a constant $d$ such that every labeled graph $(G,\ell)$ is at Hamming distance at most $d$ from $\cL$. 

\begin{theorem}\label{prop:bad-pls}
Let $\cL$ be a distributed language. Unless $\cL$ is trivially approximable, there exists a proof-labeling scheme for $\cL$ that is \emph{not} error-sensitive.  
\end{theorem}

\begin{proof}
Let $\cL$ be a non trivially approximable distributed language. Given a labeled graph $(G,\ell)\in \cL$, let $T$ be a spanning tree of $G$. It is folklore (cf., e.g., \cite{AwerbuchPV91,KormanKP10}) that  $T$ can be certified by a proof-labeling scheme where the certificate assigned to each node $u$ consists of a pair $(I(u),d(u))$ where $I(u)$ is the ID of a node $r$ picked as the root of $T$, and $d(u)$ the hop-distance in $T$ from $u$ to $r$. The verifier checks the distances the same way as it does in the proof of Theorem~\ref{prop:acyclic} (which guarantees the absence of cycles). In addition, every node checks that it agrees with its neighbors in the graph about the ID of the root (which guarantees that $T$ is not a forest with more than one tree).  At every node, if all these tests are passed at that node, then it accepts, else it rejects. 

We now prove that every proof-labeling scheme $(\prover,\verif)$ for $\cL$ can be transformed into a proof-labeling scheme $(\prover',\verif')$  for $\cL$ which is not error-sensitive. 
On a legal instance $(G,\ell)\in \cL$, the prover $\prover'$ selects a spanning tree $T$ of $G$, and provides every node $u$ with:
\begin{enumerate}
\item the certificate that the prover $\prover$ would assign to $u$ for $(G,\ell)$, denoted by~$\prover(u)$;
\item the local description of the tree $T$, together with the corresponding certificate;
\item a boolean $b(u)$, set to $true$.
\end{enumerate} 

The verifier $\verif'$ checks the correctness of the spanning tree $T$, and rejects if it is not correct. From now on, we assume that $T$ is correct.  The verifier $\verif'$ then outputs accept or reject according to the following rules. 

\begin{enumerate}
\item At every node $u$ distinct from the root of $T$, $\verif'$ accepts if and only if one of the two conditions below is fulfilled: 
\begin{enumerate}
\item $b(u)=false$, and either $\verif$ rejects at $u$, or a child $v$ of $u$ in $T$ satisfies $b(v)=false$;
\item $b(u)=true$, $\verif$ accepts at $u$, and $b(v)=true$ for every child $v$ of $u$ in $T$. 
\end{enumerate}
\item At the root of $T$, the verifier $\verif'$ rejects if and only if 
\begin{enumerate}
\item $\verif$ rejects, or a child $v$ of $u$ satisfies $b(v)=false$. 
\end{enumerate}
\end{enumerate}

By construction, if $(G,\ell)\in \cL$ then $\verif'$ accepts at all nodes, when provided with the appropriate certificates, because, with these certificates, all booleans $b$ are $true$, and $\verif$ accepts at all nodes.

 If $(G,\ell)\notin \cL$, then $\verif'$ rejects in at least one node if the given certificates do not encode a spanning tree $T$. Therefore, let us assume that the given certificates correctly encode a spanning tree $T$, rooted at $r$. Since $(G,\ell)\notin \cL$, there exists at least one node where $\verif$ rejects. Let $u$ be a node where $\verif$ rejects, such that $\verif$ rejects at no other nodes on the shortest path from $u$ to $r$ in $T$. If $u=r$, then, since $\verif$ rejects, we get that $\verif'$ rejects as well. So, let us assume that $u\neq r$. Let $u_0,u_1,\dots,u_t$ with $u_0=u$, $t\geq 1$, and $u_t=r$ be the shortest path from $u$ to $r$ in $T$. For $\verif'$ to accept at $u_0$, it must be the case that $b(u)=false$. The same holds at each node along the path: For $\verif'$ to accept at $u_i$, $i=0,\dots,t-1$, it must be the case that $b(u_i)=false$. This leads $\verif'$ to reject at $u_t=r$. Therefore, $(\prover',\verif')$ is a proof-labeling scheme for $\cL$. 

We now show that $(\prover',\verif')$ is not error-sensitive. Let $(G,\ell)\notin \cL$.  Let $T$ be a spanning tree of $G$, rooted at node $r$. We provide the nodes with the proper description of $T$ and the certificates to certify $T$. We also provide the nodes with arbitrary certificates for $\verif$. Then we provide the nodes with the following ``fake'' boolean certificates that we assign by visiting the nodes of the tree $T$ bottom-up, as follows. Let $u$ be a node:  

\begin{enumerate}
\item if $\verif$ rejects at $u$ or a child $v$ of $u$ in $T$ satisfies $b(v)=false$, then set $b(u)=false$; 
\item else set $b(u)=true$.  
\end{enumerate}

In this way, only the root of $T$ can reject. Therefore, with such certificates, even instances $(G,\ell)$ that are arbitrarily far from $\cL$ will be rejected by a single node. It follows that $(\prover',\verif')$ is not error-sensitive, as claimed. 
\end{proof}

Recall that the fact that every distributed language has a proof-labeling scheme can be established by using a \emph{universal} proof-labeling scheme $(\prover_{univ},\verif_{univ})$ (see \cite{GoosS16}). 
Given a distributed language $\cL$,
the universal proof-labeling scheme is defined as follows. 
On a legal instance $(G,\ell)\in \cL$, where $G$ has $n$ vertices, the prover assigns a certificate $c(u)=(T,M,L)$ to every node $u$. 
Specifically, the prover orders the vertices from $1$ to $n$ arbitrarily, and $T$ is a vector with $n$~entries indexed from~1 to~$n$ where $T[i]$ is the ID of the $i$-th node $u$.
Then, $L[i]$ is the label $\ell(v)$ of the $i$-th node $u$.
Finally, $M$ is the adjacency matrix of~$G$, where the $i$-th raw (and $i$-th column) corresponds to the $i$-th vertex in $T$. 
The prover $\prover_{univ}$ assigns $c(u)$ to every node $u\in V(G)$. The verifier $\verif_{univ}$ then checks at every node $u$ that its certificate is consistent with the certificates given to its neighbors (i.e., they all have the same $T$, $L$, and $M$, the indexes match with the IDs, and the actual neighborhood of $v$ is as it is specified in $T$, $L$ and $M$). 
If this test is not passed, then $\verif_{univ}$ outputs \emph{reject} at~$u$, otherwise it outputs \emph{accept} or \emph{reject} according to whether the labeled graph described by $(M,L)$ is in $\cL$ or not. It is easy to check that $(\prover_{univ},\verif_{univ})$ is indeed a proof-labeling scheme for $\cL$.  

The universal scheme uses large certificates, of size $O(n (\log n+\max_v|\ell_v|) + n^2)$. We are interested in the design of  proof-labeling schemes using significantly smaller certificates.

The universal proof-labeling scheme has the following nice property, that we state as a lemma for further references in the text.  

\begin{lemma}\label{lem:univ}
If a distributed language $\cL$ has an error-sensitive proof-labeling scheme, then the universal proof-labeling scheme applied to $\cL$ is error-sensitive. 
\end{lemma}

\begin{proof}
Let $(\prover,\verif)$ be an error-sensitive proof-labeling scheme for $\cL$, and let $(\prover_{univ},\verif_{univ})$ be the universal proof-labeling scheme for $\cL$. Let $(G,\ell)\notin \cL$. We show that $(\prover_{univ},\verif_{univ})$ is at least as good as $(\prover,\verif)$ with respect to the number of rejecting nodes. Specifically, we show that if $\verif_{univ}$ rejects $(G,\ell)$ at $r$ nodes for some certificate function $c$, then there exists a certificate function $c'$ such that $\verif$ rejects $(G,\ell)$ in at most $r$ nodes. 
We now describe how to construct the certificate assignment $c'$, on $(G,\ell)$.
Given any node~$u$, the definition of the certificate $c'(u)$ depends on the behavior of $u$  and its neighbors in the universal scheme with certificate $c$. 

\begin{itemize}
\item The first case is when $\verif_{univ}$ accepts at $u$, with certificate $c(u)=(T,M,L)$.
Let $G_{M}$ be the graph described by~$M$, and $\ell_{L}$ be the labeling of $G_{M}$ that correspond to $L$.
Since $u$ accepts, we have $(G_{M},\ell_{L})$ is in $\cL$.
The certificate $c'(u)$ is the certificate that the prover $\prover$ would assign to~$u$ if it were in $(G_{M},\ell_{L})$.

\item The second case is when: 
(1) $\verif_{univ}$ rejects $(G,\ell)$ at $u$, and 
(2) $u$ is adjacent to at least one node $v$ at which $\verif_{univ}$ accepts $(G,\ell)$. 
Note that this situation can occur only under special circumstances. 
The fact that $v$ accepts means that $u$ and $v$ were given the same triplet $(T,M,L)$, and that this triplet corresponds to a correct instance of the language. 
Therefore, the fact that $u$~rejects can only come from the fact that its neighborhood does not match the description of this neighborhood in $(T,M,L)$.
As before, we set $c'(u)$ as the certificate assigned to node $u$ by $\prover$ in the labeled graph $(G_{M},\ell_{L})$. Note that if $u$ is adjacent to two different nodes $v$ and $v'$ at which $\verif_{univ}$ accepts, then these two nodes $v$ and $v'$ share the same certificates $(T,M,L)$. Hence the definition of $c'$ at $u$ is well defined. 

\item The third case is when none of the  previous two cases apply. In this case $c'(u)$ is set to~$\varnothing$.
\end{itemize}

Let us now consider the behavior of $\verif$ on $(G,\ell) $ with certificates $c'$.
We observe that for a node $u$ in which $\verif_{univ}$ accepts, its certificate $c(u)$ is consistent with the certificates of all its neighbors, and thus, in particular, $u$ and its neighbors share the same labeled graph representation $(M,L)$. Therefore, the certificates $c'$ assigned to $u$ and its neighbors are consistent with respect to $\verif$. It follows that every node $u$ at which $\verif_{univ}$ accepts $(G,\ell)$ with certificate function $c$ satisfies that $\verif$ accepts $(G,\ell)$ at $u$ with certificate function $c'$.  
This implies the Lemma.
\end{proof}

While every distributed language has a proof-labeling scheme, we show, using Lemma~\ref{lem:univ}, that there exist languages for which there are no error-sensitive proof-labeling schemes. 

\begin{theorem}\label{prop:notallwork}
There exist languages that do not admit any error-sensitive proof-labeling scheme. 
\end{theorem}

\begin{proof}
We show that there exist languages $\cL$ such that, for every proof-labeling scheme $(\prover,\verif)$ for $\cL$, and every $d \geq 1$, there exists a labeled graph $(G,\ell)$ at Hamming distance at least $d$ from $\cL$, and a certificate function $c$, such that $\verif$ rejects $(G,\ell)$ with certificate $c$ in at most a constant number of nodes. 
We consider labeled graphs $(G,\ell)$ where $\ell$ encodes a subgraph of $G$ as follows. The label $\ell(u)$ of node~$u$ is a list of neighbors of $u$ in $G$, such that
\[
\mbox{$v$ is in the list of $\ell(u) \iff u$ is in the list of $\ell(v)$}. 
\]
Such a labeling defines a subgraph of $G$ where every edge $\{u,v\}$ of $G$ is in that subgraph if and only if $v$ is in the list of $\ell(u)$. For a given $(G,\ell)$, we define $H_{\ell}$ as the subgraph described by $\ell$.
Now, let us consider the language 
\[
\mbox{\sc regular}=\{(G,\ell): \ell \mbox{ describes a subgraph } H_\ell,  \mbox{ and } H_\ell \; \mbox{is regular}\}.
\]
Let us assume, for the purpose of contradiction, that there exists an error-sensitive proof-labeling scheme $(\prover,\verif)$ for {\sc regular}. From Lemma~\ref{lem:univ}, it follows that the universal scheme $(\prover_{univ},\verif_{univ})$ is error-sensitive for  {\sc regular}. We show that this is not the case. 

Let $d_1$ and $d_2$ be two distinct integers. Let $G_1$ be a regular graph of degree $d_1$, and let $G'_1$ be a copy of $G_1$. 
Let $\{u_1,v_1\}\in E(G_1)$, and let $\{u'_1,v'_1\}$ be the corresponding edge in $G'_1$. 
We construct the graph $G_1^*$, obtained from $G_1$ and $G'_1$, by removing $\{u_1,v_1\}$ and $\{u'_1,v'_1\}$, and adding $\{u_1,u'_1\}$ and $\{v_1,v'_1\}$. By construction,  $G_1^*$ is $d_1$-regular. Similarly, we can construct a $d_2$-regular graph $G_2^*$ from a $d_2$-regular graph $G_2$ and its copy $G'_2$. We denote by $\{u_2,u'_2\}$ and $\{v_2,v'_2\}$ the edges connecting $G_2$ to its copy $G'_2$ in $G^*_2$. For $i\in\{1,2\}$, let $\ell_i$ be the labeling of the nodes  of $G^*_i$ such that $H_{\ell_i}=G_ i^*$. We have 
\[
(G^*_1,\ell_1)\in \mbox{\sc regular}, \; \mbox{and} \; (G^*_2,\ell_2)\in \mbox{\sc regular}.
\]
Let $G^*_3$ be the graph obtained from $G_1$ and $G_2$ by removing $\{u_1,v_1\}$ from $G_1$, removing $\{u_2,v_2\}$ from $G_2$, and adding the edges $\{u_1,u_2\}$ and $\{v_1,v_2\}$. 
Again, let us consider the labels $\ell_3$ assigned to the nodes of $G^*_3$ with $H_{\ell_3}=G_ 3^*$. Since $d_1\neq d_2$, we have 
\[
(G^*_3,\ell_3)\notin \mbox{\sc regular}. 
\]
Now, let us assign to the nodes of $G_1$ in $G^*_3$ the certificates assigned by $\prover_{univ}$ to the nodes of $G_1$ in $G^*_1$. 
Similarly, let us assign to the nodes of $G'_2$ in $G^*_3$ the certificates assigned by $\prover_{univ}$ to the nodes of $G'_2$ in $G^*_2$. 
With such certificates, only the four nodes $u_1$, $v_1$, $u_2$, and $v_2$, can detect an inconsistency between their certificates and the certificates of their neighborhoods. Therefore only these nodes may reject when running $\verif_{univ}$. Therefore,  at most $4$ nodes reject. On the other hand the distance between $(G^*_3,\ell_3)$ and {\sc regular} is at least as large as $\min\{|V(G_1)|,|V(G_2)|\}$. This distance can be made arbitrarily large, while the number of rejecting nodes remains constant. Hence, the universal proof-labeling scheme is not error-sensitive.
\end{proof}

\noindent\textbf{Remark.} The language {\sc regular} used in the proof of Theorem~\ref{prop:notallwork} to establish the existence of  languages that do not admit any error-sensitive proof-labeling schemes actually belongs to the class LCL of locally checkable labelings~\cite{NaorS95}. Therefore, the fact that a language is easy to check locally does not help for the design of error-sensitive proof-labeling schemes. 

\medskip

We complete this warmup section by some observations regarding the encoding of distributed data structures. Let us consider the following two distributed languages, both corresponding to spanning tree. The first language, $\mbox{\sc st}_p$, encodes the spanning trees using pointers to parents, while the second language, $\mbox{\sc st}_l$, encodes the spanning trees by listing all the incident edges of each node in these tree. 
\begin{eqnarray*}
\mbox{\sc st}_p & = \Big\{(G,\ell) :& \forall v\in V(G), \, \ell(v)\in N(v)\cup\{\bot\} \;\\
& &\mbox{and} \; \Big ( \bigcup_{v\in V(G)\,:\, \ell(v)\neq \bot}(v,\ell(v)) \Big )\; \mbox{forms a spanning tree} \Big\} 
\end{eqnarray*}
\begin{eqnarray*}
\mbox{\sc st}_l & =  \Big\{(G,\ell) :& \forall v\in V(G), \, \ell(v)\subseteq N(v) \; \mbox{and} \; u\in \ell(v) \mbox{ iff }  v \in \ell(u), \\
& & \mbox{ and } \Big ( \bigcup_{v\in V(G)}\bigcup_{u\in \ell(v)} (u,v)\Big ) \; \mbox{forms a spanning tree} \Big\}.
\end{eqnarray*}
Obviously, $\mbox{\sc st}_p$ is just a compressed version of $\mbox{\sc st}_l$ as the latter can be constructed from the former in just one round. However, note that  $\mbox{\sc st}_p$ cannot be recover from $\mbox{\sc st}_l$ in a constant number of rounds, because $\mbox{\sc st}_p$ provides a consistent orientation of the edges in the tree. It follows that  $\mbox{\sc st}_p$ is an  encoding of spanning trees which is actually strictly richer than $\mbox{\sc st}_l$. This difference between  $\mbox{\sc st}_p$ and  $\mbox{\sc st}_l$ is not anecdotal, as  we shall prove later that $\mbox{\sc st}_l$ admits an error-sensitive proof-labeling scheme, while we show hereafter that  $\mbox{\sc st}_p$ is not appropriate for the design of error-sensitive proof-labeling schemes.  

\begin{theorem}\label{prop:STnot}
$\mbox{\sc st}_p$ does not admit any error-sensitive proof-labeling scheme. 
\end{theorem}

\begin{proof}
In this proof, we will write $\ell(u)=v$ to denote the fact that the pointer encoded in the label of $u$ is pointing towards node~$v$.
Let $P_n$ be the $n$-node path $u_1,u_2,\dots,u_n$ with $n$ even. Let $\ell_0, \ell_1$, and $\ell_2$  be labelings defined~by: 

\begin{itemize}
\item 
$\ell_1(u_i)=u_{i+1}$ for all $1 \leq i < n$, and $\ell_1(u_n)=\bot$; 
 \item 
$\ell_2(u_i)=u_{i-1}$ for all $1<i\leq n $, and $\ell_2(u_1)=\bot$; 
\item 
and $\ell_3(u_i)=u_{i-1}$ for all $1<i\leq \frac{n}{2}$, $\ell_3(u_i)=u_{i+1}$ for all $\frac{n}{2}+1\leq i < n$, and $\ell_3(u_1)=\ell_3(u_n)=\bot$. 
\end{itemize} 

We have $(P_n,\ell_1)\in \mbox{\sc st}_p$ and $(P_n,\ell_2)\in \mbox{\sc st}_p$. 
The distance from $(P_n,\ell_3)$ to $\mbox{\sc st}_p$ is at least~$\frac{n}{2}$.
Indeed, let us modifying $(P_n,\ell_3)$ to get a correct instance $(P_n,\ell_4)$. 
Suppose, w.l.o.g., that the root of the tree described by $\ell_4$ is among the first half of the nodes. Then, to get from $\ell_3$ to $\ell_4$, all the pointers of the second half have to be changed, which means that the certificates in at least $n/2$ nodes must be modified.

 Let $(\prover,\verif)$ be a proof-labeling scheme for $\mbox{\sc st}_p$. Consider the case of $(P_n,\ell_3)$ where every $u_i$, $i=1,\dots,\frac{n}{2}$, is given the certificate assigned by $\prover$ to $u_i$ in $(P_n,\ell_2)$, and every $u_i$, $i=\frac{n}{2}+1,\dots,n$, is given the certificate assigned by $\prover$ to $u_i$ in $(P_n,\ell_1)$.  
With such certificates, all nodes $u_i$ for $i=1,\dots, u_{\frac{n}{2}-1}$ have the exact same view as in $(P_n,\ell_1)$, and all nodes $u_i$ for $i=u_{\frac{n}{2}+2},\dots, n$ have the exact same view as in $(P_n,\ell_2)$. Therefore all these $n$~nodes must accept.
Hence, $(P_n,\ell_3)$ can only be rejected by $\verif$ at the two nodes $u_{\frac{n}{2}}$ and~$u_{\frac{n}{2}+1}$.
\end{proof}


\section{Characterization}
\label{sec:characterization}

In this section, we define a notion of \emph{local stability} for languages (Definition~\ref{def:local-stability}), and show that being locally stable is equivalent to the fact of having an error-sensitive scheme (Theorem~\ref{thm:characterization}).
Then, we discuss a simpler but less general version of local stability, that we call \emph{strong local stability}.
Finally, we give several examples of application of our equivalence therorem. 

Roughly, local stability captures whether a patchwork of several correct instances (with a small contact area between the instances), can be a ``very incorrect'' instance, or an ``almost correct'' instance. 
For example, the language \mbox{\sc regular} from the previous section in non locally stable, because, by gluing together two regular graphs, one can get a graph that is very far from being regular whenever the original graphs have different degrees. 


In order to define the notions of local stability, we need to formalize the notion of a ``patchwork of solutions'' and of ``contact area''.
Let $G$ be a graph, and let $H$ be a subgraph of $G$, that is, a graph $H$ such that $V(H)\subseteq V(G)$, and $E(H)\subseteq E(G)$. We first define partial labelings and induced labelings (see Figure~\ref{fig:induced-labeling}).

\begin{definition}[Partial labeling]
 Given a labeling $\ell$ of a graph $G$, and a subgraph $H$ of $G$, the \emph{partial labeling} $\ell_H$ denotes the labeling of $H$ induced by $\ell$ restricted to the nodes of $H$: 
\[
\ell_H(v)=\left \{ \begin{array}{ll}
\ell(v) & \mbox{if $v\in V(H)$} \\
\varnothing & \mbox{otherwise (where $\varnothing $ denotes the empty string).}
\end{array}\right.
\]
\end{definition}

\begin{definition}[Induced labeling] 
Let $G$ be a graph, and let $H_1,\dots,H_k$ be a family of connected subgraphs of $G$ such that $(V(H_i))_{i=1,\dots,k}$ is a partition of $V(G)$. For every $i\in\{1,\dots,k\}$, let us consider a labeled graph $(G_i,\ell_i)\in \cL$ such that $H_i$ is a subgraph of $G_i$. 
Let $\ell$ be the following labeling of $G$: for every $v\in V(G)$, $\ell(v)=\ell_i(v)$ where $i$ is such that $v\in V(H_i)$. We say that such a labeled graph $(G,\ell)$ is \emph{induced} by the labeled graphs $(G_i,\ell_i)$, $i=1,\dots,k$, \emph{via} the subgraphs $H_1,\dots,H_k$. 
\end{definition}

\begin{figure}[tb]
\begin{tabular}{cc}
\begin{minipage}{0.55\textwidth}
\scalebox{0.83}{
\tikzset{every picture/.style={line width=0.75pt}} 

\begin{tikzpicture}[x=0.75pt,y=0.75pt,yscale=-1,xscale=1]

\draw [line width=1.5]  [dash pattern={on 5.63pt off 4.5pt}]  (120.77,95.73) .. controls (177.77,111.73) and (222.77,55.73) .. (221.77,20.73) ;
\draw [fill={rgb, 255:red, 255; green, 255; blue, 255 }  ,fill opacity=1 ][line width=1.5]    (137.27,51.27) -- (138.27,127.27) ;
\draw [fill={rgb, 255:red, 255; green, 255; blue, 255 }  ,fill opacity=1 ][line width=1.5]    (137.27,51.27) -- (154.27,79.27) ;
\draw [fill={rgb, 255:red, 255; green, 255; blue, 255 }  ,fill opacity=1 ][line width=1.5]    (137.27,51.27) -- (180.27,48.27) ;
\draw [fill={rgb, 255:red, 255; green, 255; blue, 255 }  ,fill opacity=1 ][line width=1.5]    (180.27,48.27) -- (203.27,28.27) ;
\draw [fill={rgb, 255:red, 255; green, 255; blue, 255 }  ,fill opacity=1 ][line width=1.5]    (180.27,48.27) -- (232.27,49.27) ;
\draw [fill={rgb, 255:red, 255; green, 255; blue, 255 }  ,fill opacity=1 ][line width=1.5]    (245.27,108.27) -- (232.27,49.27) ;
\draw [fill={rgb, 255:red, 255; green, 255; blue, 255 }  ,fill opacity=1 ][line width=1.5]    (203.27,92.27) -- (232.27,49.27) ;
\draw [fill={rgb, 255:red, 255; green, 255; blue, 255 }  ,fill opacity=1 ][line width=1.5]    (203.27,92.27) -- (245.27,108.27) ;
\draw [fill={rgb, 255:red, 255; green, 255; blue, 255 }  ,fill opacity=1 ][line width=1.5]    (154.27,79.27) -- (203.27,92.27) ;
\draw [fill={rgb, 255:red, 255; green, 255; blue, 255 }  ,fill opacity=1 ][line width=1.5]    (166.27,115.27) -- (203.27,92.27) ;
\draw [fill={rgb, 255:red, 255; green, 255; blue, 255 }  ,fill opacity=1 ][line width=1.5]    (203.27,28.27) -- (232.27,49.27) ;
\draw [fill={rgb, 255:red, 255; green, 255; blue, 255 }  ,fill opacity=1 ][line width=1.5]    (137.27,51.27) -- (203.27,28.27) ;
\draw [fill={rgb, 255:red, 255; green, 255; blue, 255 }  ,fill opacity=1 ][line width=1.5]    (154.27,79.27) -- (180.27,48.27) ;
\draw [fill={rgb, 255:red, 255; green, 255; blue, 255 }  ,fill opacity=1 ][line width=1.5]    (138.27,127.27) -- (154.27,79.27) ;
\draw [fill={rgb, 255:red, 255; green, 255; blue, 255 }  ,fill opacity=1 ][line width=1.5]    (138.27,127.27) -- (188.27,159.27) ;
\draw [fill={rgb, 255:red, 255; green, 255; blue, 255 }  ,fill opacity=1 ][line width=1.5]    (218.27,128.27) -- (188.27,159.27) ;
\draw [fill={rgb, 255:red, 255; green, 255; blue, 255 }  ,fill opacity=1 ][line width=1.5]    (245.27,108.27) -- (218.27,128.27) ;
\draw [fill={rgb, 255:red, 255; green, 255; blue, 255 }  ,fill opacity=1 ][line width=1.5]    (218.27,128.27) -- (166.27,115.27) ;
\draw [fill={rgb, 255:red, 255; green, 255; blue, 255 }  ,fill opacity=1 ][line width=1.5]    (166.27,115.27) -- (154.27,79.27) ;
\draw  [fill={rgb, 255:red, 255; green, 255; blue, 255 }  ,fill opacity=1 ][line width=1.5]  (130.77,127.27) .. controls (130.77,123.12) and (134.12,119.77) .. (138.27,119.77) .. controls (142.41,119.77) and (145.77,123.12) .. (145.77,127.27) .. controls (145.77,131.41) and (142.41,134.77) .. (138.27,134.77) .. controls (134.12,134.77) and (130.77,131.41) .. (130.77,127.27) -- cycle ;
\draw  [fill={rgb, 255:red, 255; green, 255; blue, 255 }  ,fill opacity=1 ][line width=1.5]  (180.77,159.27) .. controls (180.77,155.12) and (184.12,151.77) .. (188.27,151.77) .. controls (192.41,151.77) and (195.77,155.12) .. (195.77,159.27) .. controls (195.77,163.41) and (192.41,166.77) .. (188.27,166.77) .. controls (184.12,166.77) and (180.77,163.41) .. (180.77,159.27) -- cycle ;
\draw  [fill={rgb, 255:red, 255; green, 255; blue, 255 }  ,fill opacity=1 ][line width=1.5]  (210.77,128.27) .. controls (210.77,124.12) and (214.12,120.77) .. (218.27,120.77) .. controls (222.41,120.77) and (225.77,124.12) .. (225.77,128.27) .. controls (225.77,132.41) and (222.41,135.77) .. (218.27,135.77) .. controls (214.12,135.77) and (210.77,132.41) .. (210.77,128.27) -- cycle ;
\draw  [fill={rgb, 255:red, 255; green, 255; blue, 255 }  ,fill opacity=1 ][line width=1.5]  (237.77,108.27) .. controls (237.77,104.12) and (241.12,100.77) .. (245.27,100.77) .. controls (249.41,100.77) and (252.77,104.12) .. (252.77,108.27) .. controls (252.77,112.41) and (249.41,115.77) .. (245.27,115.77) .. controls (241.12,115.77) and (237.77,112.41) .. (237.77,108.27) -- cycle ;
\draw  [fill={rgb, 255:red, 255; green, 255; blue, 255 }  ,fill opacity=1 ][line width=1.5]  (195.77,92.27) .. controls (195.77,88.12) and (199.12,84.77) .. (203.27,84.77) .. controls (207.41,84.77) and (210.77,88.12) .. (210.77,92.27) .. controls (210.77,96.41) and (207.41,99.77) .. (203.27,99.77) .. controls (199.12,99.77) and (195.77,96.41) .. (195.77,92.27) -- cycle ;
\draw  [fill={rgb, 255:red, 255; green, 255; blue, 255 }  ,fill opacity=1 ][line width=1.5]  (158.77,115.27) .. controls (158.77,111.12) and (162.12,107.77) .. (166.27,107.77) .. controls (170.41,107.77) and (173.77,111.12) .. (173.77,115.27) .. controls (173.77,119.41) and (170.41,122.77) .. (166.27,122.77) .. controls (162.12,122.77) and (158.77,119.41) .. (158.77,115.27) -- cycle ;
\draw  [fill={rgb, 255:red, 255; green, 255; blue, 255 }  ,fill opacity=1 ][line width=1.5]  (224.77,49.27) .. controls (224.77,45.12) and (228.12,41.77) .. (232.27,41.77) .. controls (236.41,41.77) and (239.77,45.12) .. (239.77,49.27) .. controls (239.77,53.41) and (236.41,56.77) .. (232.27,56.77) .. controls (228.12,56.77) and (224.77,53.41) .. (224.77,49.27) -- cycle ;
\draw  [fill={rgb, 255:red, 255; green, 255; blue, 255 }  ,fill opacity=1 ][line width=1.5]  (195.77,28.27) .. controls (195.77,24.12) and (199.12,20.77) .. (203.27,20.77) .. controls (207.41,20.77) and (210.77,24.12) .. (210.77,28.27) .. controls (210.77,32.41) and (207.41,35.77) .. (203.27,35.77) .. controls (199.12,35.77) and (195.77,32.41) .. (195.77,28.27) -- cycle ;
\draw  [fill={rgb, 255:red, 255; green, 255; blue, 255 }  ,fill opacity=1 ][line width=1.5]  (172.77,48.27) .. controls (172.77,44.12) and (176.12,40.77) .. (180.27,40.77) .. controls (184.41,40.77) and (187.77,44.12) .. (187.77,48.27) .. controls (187.77,52.41) and (184.41,55.77) .. (180.27,55.77) .. controls (176.12,55.77) and (172.77,52.41) .. (172.77,48.27) -- cycle ;
\draw  [fill={rgb, 255:red, 255; green, 255; blue, 255 }  ,fill opacity=1 ][line width=1.5]  (129.77,51.27) .. controls (129.77,47.12) and (133.12,43.77) .. (137.27,43.77) .. controls (141.41,43.77) and (144.77,47.12) .. (144.77,51.27) .. controls (144.77,55.41) and (141.41,58.77) .. (137.27,58.77) .. controls (133.12,58.77) and (129.77,55.41) .. (129.77,51.27) -- cycle ;
\draw  [fill={rgb, 255:red, 255; green, 255; blue, 255 }  ,fill opacity=1 ][line width=1.5]  (146.77,79.27) .. controls (146.77,75.12) and (150.12,71.77) .. (154.27,71.77) .. controls (158.41,71.77) and (161.77,75.12) .. (161.77,79.27) .. controls (161.77,83.41) and (158.41,86.77) .. (154.27,86.77) .. controls (150.12,86.77) and (146.77,83.41) .. (146.77,79.27) -- cycle ;

\end{tikzpicture}
}
\scalebox{0.83}{
\tikzset{every picture/.style={line width=0.75pt}} 

\begin{tikzpicture}[x=0.75pt,y=0.75pt,yscale=-1,xscale=1]

\draw [line width=1.5]    (296.27,133.27) -- (345.27,161.27) ;
\draw [line width=1.5]    (390.27,52.27) -- (361.27,96.27) ;
\draw [line width=1.5]    (324.27,118.27) -- (376.27,133.27) ;
\draw [line width=1.5]    (324.27,118.27) -- (361.27,96.27) ;
\draw [line width=1.5]    (310.27,82.27) -- (361.27,96.27) ;
\draw [line width=1.5]    (324.27,118.27) -- (345.27,161.27) ;
\draw [line width=1.5]    (310.27,82.27) -- (324.27,118.27) ;
\draw [line width=1.5]    (310.27,82.27) -- (339.27,53.27) ;
\draw [line width=1.5]    (390.27,52.27) -- (426.27,53.27) ;
\draw [line width=1.5]    (339.27,53.27) -- (390.27,52.27) ;
\draw [line width=1.5]    (309.27,24.27) -- (390.27,52.27) ;
\draw [line width=1.5]    (310.27,82.27) -- (296.27,133.27) ;
\draw [line width=1.5]    (295.27,52.27) -- (310.27,82.27) ;
\draw [line width=1.5]    (295.27,52.27) -- (339.27,53.27) ;
\draw [line width=1.5]    (309.27,24.27) -- (339.27,53.27) ;
\draw [line width=1.5]    (426.27,53.27) -- (440.27,89.27) ;
\draw [line width=1.5]    (390.27,52.27) -- (405.27,109.27) ;
\draw [line width=1.5]    (405.27,109.27) -- (440.27,89.27) ;
\draw [line width=1.5]    (404.27,146.27) -- (440.27,89.27) ;
\draw [line width=1.5]    (405.27,109.27) -- (404.27,146.27) ;
\draw [line width=1.5]    (345.27,161.27) -- (404.27,146.27) ;
\draw [line width=1.5]    (345.27,161.27) -- (376.27,133.27) ;
\draw [line width=1.5]    (376.27,133.27) -- (405.27,109.27) ;
\draw [line width=1.5]    (361.27,96.27) -- (405.27,109.27) ;
\draw  [fill={rgb, 255:red, 190; green, 190; blue, 190 }  ,fill opacity=1 ][line width=1.5]  (301.77,24.27) .. controls (301.77,20.12) and (305.12,16.77) .. (309.27,16.77) .. controls (313.41,16.77) and (316.77,20.12) .. (316.77,24.27) .. controls (316.77,28.41) and (313.41,31.77) .. (309.27,31.77) .. controls (305.12,31.77) and (301.77,28.41) .. (301.77,24.27) -- cycle ;
\draw  [fill={rgb, 255:red, 190; green, 190; blue, 190 }  ,fill opacity=1 ][line width=1.5]  (382.77,52.27) .. controls (382.77,48.12) and (386.12,44.77) .. (390.27,44.77) .. controls (394.41,44.77) and (397.77,48.12) .. (397.77,52.27) .. controls (397.77,56.41) and (394.41,59.77) .. (390.27,59.77) .. controls (386.12,59.77) and (382.77,56.41) .. (382.77,52.27) -- cycle ;
\draw  [fill={rgb, 255:red, 190; green, 190; blue, 190 }  ,fill opacity=1 ][line width=1.5]  (418.77,53.27) .. controls (418.77,49.12) and (422.12,45.77) .. (426.27,45.77) .. controls (430.41,45.77) and (433.77,49.12) .. (433.77,53.27) .. controls (433.77,57.41) and (430.41,60.77) .. (426.27,60.77) .. controls (422.12,60.77) and (418.77,57.41) .. (418.77,53.27) -- cycle ;
\draw  [fill={rgb, 255:red, 190; green, 190; blue, 190 }  ,fill opacity=1 ][line width=1.5]  (432.77,89.27) .. controls (432.77,85.12) and (436.12,81.77) .. (440.27,81.77) .. controls (444.41,81.77) and (447.77,85.12) .. (447.77,89.27) .. controls (447.77,93.41) and (444.41,96.77) .. (440.27,96.77) .. controls (436.12,96.77) and (432.77,93.41) .. (432.77,89.27) -- cycle ;
\draw  [fill={rgb, 255:red, 190; green, 190; blue, 190 }  ,fill opacity=1 ][line width=1.5]  (397.77,109.27) .. controls (397.77,105.12) and (401.12,101.77) .. (405.27,101.77) .. controls (409.41,101.77) and (412.77,105.12) .. (412.77,109.27) .. controls (412.77,113.41) and (409.41,116.77) .. (405.27,116.77) .. controls (401.12,116.77) and (397.77,113.41) .. (397.77,109.27) -- cycle ;
\draw  [fill={rgb, 255:red, 190; green, 190; blue, 190 }  ,fill opacity=1 ][line width=1.5]  (353.77,96.27) .. controls (353.77,92.12) and (357.12,88.77) .. (361.27,88.77) .. controls (365.41,88.77) and (368.77,92.12) .. (368.77,96.27) .. controls (368.77,100.41) and (365.41,103.77) .. (361.27,103.77) .. controls (357.12,103.77) and (353.77,100.41) .. (353.77,96.27) -- cycle ;
\draw  [fill={rgb, 255:red, 190; green, 190; blue, 190 }  ,fill opacity=1 ][line width=1.5]  (368.77,133.27) .. controls (368.77,129.12) and (372.12,125.77) .. (376.27,125.77) .. controls (380.41,125.77) and (383.77,129.12) .. (383.77,133.27) .. controls (383.77,137.41) and (380.41,140.77) .. (376.27,140.77) .. controls (372.12,140.77) and (368.77,137.41) .. (368.77,133.27) -- cycle ;
\draw  [fill={rgb, 255:red, 190; green, 190; blue, 190 }  ,fill opacity=1 ][line width=1.5]  (337.77,161.27) .. controls (337.77,157.12) and (341.12,153.77) .. (345.27,153.77) .. controls (349.41,153.77) and (352.77,157.12) .. (352.77,161.27) .. controls (352.77,165.41) and (349.41,168.77) .. (345.27,168.77) .. controls (341.12,168.77) and (337.77,165.41) .. (337.77,161.27) -- cycle ;
\draw  [fill={rgb, 255:red, 190; green, 190; blue, 190 }  ,fill opacity=1 ][line width=1.5]  (316.77,118.27) .. controls (316.77,114.12) and (320.12,110.77) .. (324.27,110.77) .. controls (328.41,110.77) and (331.77,114.12) .. (331.77,118.27) .. controls (331.77,122.41) and (328.41,125.77) .. (324.27,125.77) .. controls (320.12,125.77) and (316.77,122.41) .. (316.77,118.27) -- cycle ;
\draw  [fill={rgb, 255:red, 190; green, 190; blue, 190 }  ,fill opacity=1 ][line width=1.5]  (288.77,133.27) .. controls (288.77,129.12) and (292.12,125.77) .. (296.27,125.77) .. controls (300.41,125.77) and (303.77,129.12) .. (303.77,133.27) .. controls (303.77,137.41) and (300.41,140.77) .. (296.27,140.77) .. controls (292.12,140.77) and (288.77,137.41) .. (288.77,133.27) -- cycle ;
\draw  [fill={rgb, 255:red, 190; green, 190; blue, 190 }  ,fill opacity=1 ][line width=1.5]  (302.77,82.27) .. controls (302.77,78.12) and (306.12,74.77) .. (310.27,74.77) .. controls (314.41,74.77) and (317.77,78.12) .. (317.77,82.27) .. controls (317.77,86.41) and (314.41,89.77) .. (310.27,89.77) .. controls (306.12,89.77) and (302.77,86.41) .. (302.77,82.27) -- cycle ;
\draw  [fill={rgb, 255:red, 190; green, 190; blue, 190 }  ,fill opacity=1 ][line width=1.5]  (331.77,53.27) .. controls (331.77,49.12) and (335.12,45.77) .. (339.27,45.77) .. controls (343.41,45.77) and (346.77,49.12) .. (346.77,53.27) .. controls (346.77,57.41) and (343.41,60.77) .. (339.27,60.77) .. controls (335.12,60.77) and (331.77,57.41) .. (331.77,53.27) -- cycle ;
\draw  [fill={rgb, 255:red, 190; green, 190; blue, 190 }  ,fill opacity=1 ][line width=1.5]  (287.77,52.27) .. controls (287.77,48.12) and (291.12,44.77) .. (295.27,44.77) .. controls (299.41,44.77) and (302.77,48.12) .. (302.77,52.27) .. controls (302.77,56.41) and (299.41,59.77) .. (295.27,59.77) .. controls (291.12,59.77) and (287.77,56.41) .. (287.77,52.27) -- cycle ;
\draw  [fill={rgb, 255:red, 190; green, 190; blue, 190 }  ,fill opacity=1 ][line width=1.5]  (396.77,146.27) .. controls (396.77,142.12) and (400.12,138.77) .. (404.27,138.77) .. controls (408.41,138.77) and (411.77,142.12) .. (411.77,146.27) .. controls (411.77,150.41) and (408.41,153.77) .. (404.27,153.77) .. controls (400.12,153.77) and (396.77,150.41) .. (396.77,146.27) -- cycle ;
\draw [line width=1.5]  [dash pattern={on 5.63pt off 4.5pt}]  (361,72) .. controls (379.37,52.53) and (384.37,37.53) .. (377.37,16.53) ;
\draw [line width=1.5]  [dash pattern={on 5.63pt off 4.5pt}]  (361,72) .. controls (413.37,87.53) and (361.37,166.53) .. (428.37,177.53) ;

\end{tikzpicture}}
\vspace{1cm}
\hspace{2.0cm}
\scalebox{0.83}{
\tikzset{every picture/.style={line width=0.75pt}} 

\begin{tikzpicture}[x=0.75pt,y=0.75pt,yscale=-1,xscale=1]

\draw [line width=1.5]    (193.27,145.27) -- (242.27,173.27) ;
\draw [line width=1.5]    (242.27,173.27) -- (301.27,157.27) ;
\draw [line width=1.5]    (242.27,173.27) -- (274.27,143.27) ;
\draw [line width=1.5]    (301.27,157.27) -- (338.27,101.27) ;
\draw [line width=1.5]    (301.27,121.27) -- (338.27,101.27) ;
\draw [line width=1.5]    (301.27,157.27) -- (301.27,121.27) ;
\draw [line width=1.5]    (221.27,127.27) -- (274.27,143.27) ;
\draw [line width=1.5]    (242.27,173.27) -- (221.27,127.27) ;
\draw [line width=1.5]    (193.27,145.27) -- (206.27,92.27) ;
\draw [line width=1.5]    (206.27,92.27) -- (259.27,107.27) ;
\draw [line width=1.5]    (206.27,92.27) -- (221.27,127.27) ;
\draw [line width=1.5]    (259.27,107.27) -- (301.27,121.27) ;
\draw [line width=1.5]    (221.27,127.27) -- (259.27,107.27) ;
\draw [line width=1.5]    (274.27,143.27) -- (301.27,121.27) ;
\draw  [fill={gray}  ,fill opacity=1 ][line width=1.5]  (185.77,145.27) .. controls (185.77,141.12) and (189.12,137.77) .. (193.27,137.77) .. controls (197.41,137.77) and (200.77,141.12) .. (200.77,145.27) .. controls (200.77,149.41) and (197.41,152.77) .. (193.27,152.77) .. controls (189.12,152.77) and (185.77,149.41) .. (185.77,145.27) -- cycle ;
\draw  [fill={gray}  ,fill opacity=1 ][line width=1.5]  (234.77,173.27) .. controls (234.77,169.12) and (238.12,165.77) .. (242.27,165.77) .. controls (246.41,165.77) and (249.77,169.12) .. (249.77,173.27) .. controls (249.77,177.41) and (246.41,180.77) .. (242.27,180.77) .. controls (238.12,180.77) and (234.77,177.41) .. (234.77,173.27) -- cycle ;
\draw  [fill={gray}  ,fill opacity=1 ][line width=1.5]  (293.77,157.27) .. controls (293.77,153.12) and (297.12,149.77) .. (301.27,149.77) .. controls (305.41,149.77) and (308.77,153.12) .. (308.77,157.27) .. controls (308.77,161.41) and (305.41,164.77) .. (301.27,164.77) .. controls (297.12,164.77) and (293.77,161.41) .. (293.77,157.27) -- cycle ;
\draw  [fill={gray}  ,fill opacity=1 ][line width=1.5]  (330.77,101.27) .. controls (330.77,97.12) and (334.12,93.77) .. (338.27,93.77) .. controls (342.41,93.77) and (345.77,97.12) .. (345.77,101.27) .. controls (345.77,105.41) and (342.41,108.77) .. (338.27,108.77) .. controls (334.12,108.77) and (330.77,105.41) .. (330.77,101.27) -- cycle ;
\draw  [fill={gray}  ,fill opacity=1 ][line width=1.5]  (293.77,121.27) .. controls (293.77,117.12) and (297.12,113.77) .. (301.27,113.77) .. controls (305.41,113.77) and (308.77,117.12) .. (308.77,121.27) .. controls (308.77,125.41) and (305.41,128.77) .. (301.27,128.77) .. controls (297.12,128.77) and (293.77,125.41) .. (293.77,121.27) -- cycle ;
\draw  [fill={gray}  ,fill opacity=1 ][line width=1.5]  (251.77,107.27) .. controls (251.77,103.12) and (255.12,99.77) .. (259.27,99.77) .. controls (263.41,99.77) and (266.77,103.12) .. (266.77,107.27) .. controls (266.77,111.41) and (263.41,114.77) .. (259.27,114.77) .. controls (255.12,114.77) and (251.77,111.41) .. (251.77,107.27) -- cycle ;
\draw  [fill={gray}  ,fill opacity=1 ][line width=1.5]  (266.77,143.27) .. controls (266.77,139.12) and (270.12,135.77) .. (274.27,135.77) .. controls (278.41,135.77) and (281.77,139.12) .. (281.77,143.27) .. controls (281.77,147.41) and (278.41,150.77) .. (274.27,150.77) .. controls (270.12,150.77) and (266.77,147.41) .. (266.77,143.27) -- cycle ;
\draw  [fill={gray}  ,fill opacity=1 ][line width=1.5]  (213.77,127.27) .. controls (213.77,123.12) and (217.12,119.77) .. (221.27,119.77) .. controls (225.41,119.77) and (228.77,123.12) .. (228.77,127.27) .. controls (228.77,131.41) and (225.41,134.77) .. (221.27,134.77) .. controls (217.12,134.77) and (213.77,131.41) .. (213.77,127.27) -- cycle ;
\draw  [fill={gray}  ,fill opacity=1 ][line width=1.5]  (198.77,92.27) .. controls (198.77,88.12) and (202.12,84.77) .. (206.27,84.77) .. controls (210.41,84.77) and (213.77,88.12) .. (213.77,92.27) .. controls (213.77,96.41) and (210.41,99.77) .. (206.27,99.77) .. controls (202.12,99.77) and (198.77,96.41) .. (198.77,92.27) -- cycle ;
\draw [line width=1.5]  [dash pattern={on 5.63pt off 4.5pt}]  (176.17,111.53) .. controls (218.97,114.53) and (232.37,111.53) .. (257.37,83.53) ;
\draw [line width=1.5]  [dash pattern={on 5.63pt off 4.5pt}]  (257.37,83.53) .. controls (313.17,95.53) and (265.17,184.53) .. (312.17,184.53) ;

\end{tikzpicture}}
\end{minipage}
&
\begin{minipage}{0.4\textwidth}
\scalebox{0.86}{
\tikzset{every picture/.style={line width=0.75pt}} 

\begin{tikzpicture}[x=0.75pt,y=0.75pt,yscale=-1,xscale=1]


\draw [line width=1.5]    (225.48,18.0) -- (200.00,45.0) ;

\draw [line width=1.5]    (143.48,47.48) -- (143.48,145.48) ;
\draw [line width=1.5]    (143.48,145.48) -- (206.48,181.48) ;
\draw [line width=1.5]    (206.48,181.48) -- (275.48,162.48) ;
\draw [line width=1.5]    (275.48,162.48) -- (321.48,90.48) ;
\draw [line width=1.5]    (304.48,46.48) -- (321.48,90.48) ;
\draw [line width=1.5]    (258.48,44.48) -- (304.48,46.48) ;
\draw [line width=1.5]    (223.48,20.48) -- (258.48,44.48) ;
\draw [line width=1.5]    (143.48,47.48) -- (223.48,20.48) ;
\draw [line width=1.5]    (198.48,46.48) -- (258.48,44.48) ;
\draw [line width=1.5]    (143.48,47.48) -- (198.48,46.48) ;
\draw [line width=1.5]    (198.48,46.48) -- (162.48,82.48) ;
\draw [line width=1.5]    (162.48,82.48) -- (143.48,145.48) ;
\draw [line width=1.5]    (143.48,47.48) -- (162.48,82.48) ;
\draw [line width=1.5]    (162.48,82.48) -- (222.48,100.48) ;
\draw [line width=1.5]    (258.48,44.48) -- (222.48,100.48) ;
\draw [line width=1.5]    (258.48,44.48) -- (275.48,117.48) ;
\draw [line width=1.5]    (222.48,100.48) -- (275.48,117.48) ;
\draw [line width=1.5]    (275.48,117.48) -- (321.48,90.48) ;
\draw [line width=1.5]    (275.48,117.48) -- (275.48,162.48) ;
\draw [line width=1.5]    (162.48,82.48) -- (179.48,127.48) ;
\draw [line width=1.5]    (179.48,127.48) -- (206.48,181.48) ;
\draw [line width=1.5]    (241.48,145.48) -- (275.48,117.48) ;
\draw [line width=1.5]    (206.48,181.48) -- (241.48,145.48) ;
\draw [line width=1.5]    (179.48,127.48) -- (241.48,145.48) ;
\draw [line width=1.5]    (179.48,127.48) -- (222.48,100.48) ;
\draw  [fill={white}  ,fill opacity=1 ][line width=1.5]  (215,20.48) .. controls (215,15.8) and (218.8,12) .. (223.48,12) .. controls (228.17,12) and (231.97,15.8) .. (231.97,20.48) .. controls (231.97,25.17) and (228.17,28.97) .. (223.48,28.97) .. controls (218.8,28.97) and (215,25.17) .. (215,20.48) -- cycle ;
\draw  [fill={rgb, 255:red, 190; green, 190; blue, 190 }  ,fill opacity=1 ][line width=1.5]  (250,44.48) .. controls (250,39.8) and (253.8,36) .. (258.48,36) .. controls (263.17,36) and (266.97,39.8) .. (266.97,44.48) .. controls (266.97,49.17) and (263.17,52.97) .. (258.48,52.97) .. controls (253.8,52.97) and (250,49.17) .. (250,44.48) -- cycle ;
\draw  [fill={rgb, 255:red, 190; green, 190; blue, 190 }  ,fill opacity=1 ][line width=1.5]  (296,46.48) .. controls (296,41.8) and (299.8,38) .. (304.48,38) .. controls (309.17,38) and (312.97,41.8) .. (312.97,46.48) .. controls (312.97,51.17) and (309.17,54.97) .. (304.48,54.97) .. controls (299.8,54.97) and (296,51.17) .. (296,46.48) -- cycle ;
\draw  [fill={rgb, 255:red, 190; green, 190; blue, 190 }  ,fill opacity=1 ][line width=1.5]  (313,90.48) .. controls (313,85.8) and (316.8,82) .. (321.48,82) .. controls (326.17,82) and (329.97,85.8) .. (329.97,90.48) .. controls (329.97,95.17) and (326.17,98.97) .. (321.48,98.97) .. controls (316.8,98.97) and (313,95.17) .. (313,90.48) -- cycle ;
\draw  [fill={rgb, 255:red, 190; green, 190; blue, 190 }  ,fill opacity=1 ][line width=1.5]  (267,162.48) .. controls (267,157.8) and (270.8,154) .. (275.48,154) .. controls (280.17,154) and (283.97,157.8) .. (283.97,162.48) .. controls (283.97,167.17) and (280.17,170.97) .. (275.48,170.97) .. controls (270.8,170.97) and (267,167.17) .. (267,162.48) -- cycle ;
\draw  [fill={rgb, 255:red, 190; green, 190; blue, 190 }  ,fill opacity=1 ][line width=1.5]  (267,117.48) .. controls (267,112.8) and (270.8,109) .. (275.48,109) .. controls (280.17,109) and (283.97,112.8) .. (283.97,117.48) .. controls (283.97,122.17) and (280.17,125.97) .. (275.48,125.97) .. controls (270.8,125.97) and (267,122.17) .. (267,117.48) -- cycle ;
\draw  [fill={gray}  ,fill opacity=1 ][line width=1.5]  (214,100.48) .. controls (214,95.8) and (217.8,92) .. (222.48,92) .. controls (227.17,92) and (230.97,95.8) .. (230.97,100.48) .. controls (230.97,105.17) and (227.17,108.97) .. (222.48,108.97) .. controls (217.8,108.97) and (214,105.17) .. (214,100.48) -- cycle ;
\draw  [fill={white}  ,fill opacity=1 ][line width=1.5]  (190,46.48) .. controls (190,41.8) and (193.8,38) .. (198.48,38) .. controls (203.17,38) and (206.97,41.8) .. (206.97,46.48) .. controls (206.97,51.17) and (203.17,54.97) .. (198.48,54.97) .. controls (193.8,54.97) and (190,51.17) .. (190,46.48) -- cycle ;
\draw  [fill={white}  ,fill opacity=1 ][line width=1.5]  (135,47.48) .. controls (135,42.8) and (138.8,39) .. (143.48,39) .. controls (148.17,39) and (151.97,42.8) .. (151.97,47.48) .. controls (151.97,52.17) and (148.17,55.97) .. (143.48,55.97) .. controls (138.8,55.97) and (135,52.17) .. (135,47.48) -- cycle ;
\draw  [fill={white}  ,fill opacity=1 ][line width=1.5]  (154,82.48) .. controls (154,77.8) and (157.8,74) .. (162.48,74) .. controls (167.17,74) and (170.97,77.8) .. (170.97,82.48) .. controls (170.97,87.17) and (167.17,90.97) .. (162.48,90.97) .. controls (157.8,90.97) and (154,87.17) .. (154,82.48) -- cycle ;
\draw  [fill={gray}  ,fill opacity=1 ][line width=1.5]  (135,145.48) .. controls (135,140.8) and (138.8,137) .. (143.48,137) .. controls (148.17,137) and (151.97,140.8) .. (151.97,145.48) .. controls (151.97,150.17) and (148.17,153.97) .. (143.48,153.97) .. controls (138.8,153.97) and (135,150.17) .. (135,145.48) -- cycle ;
\draw  [fill={gray}  ,fill opacity=1 ][line width=1.5]  (171,127.48) .. controls (171,122.8) and (174.8,119) .. (179.48,119) .. controls (184.17,119) and (187.97,122.8) .. (187.97,127.48) .. controls (187.97,132.17) and (184.17,135.97) .. (179.48,135.97) .. controls (174.8,135.97) and (171,132.17) .. (171,127.48) -- cycle ;
\draw  [fill={gray}  ,fill opacity=1 ][line width=1.5]  (233,145.48) .. controls (233,140.8) and (236.8,137) .. (241.48,137) .. controls (246.17,137) and (249.97,140.8) .. (249.97,145.48) .. controls (249.97,150.17) and (246.17,153.97) .. (241.48,153.97) .. controls (236.8,153.97) and (233,150.17) .. (233,145.48) -- cycle ;
\draw  [fill={gray}  ,fill opacity=1 ][line width=1.5]  (198,181.48) .. controls (198,176.8) and (201.8,173) .. (206.48,173) .. controls (211.17,173) and (214.97,176.8) .. (214.97,181.48) .. controls (214.97,186.17) and (211.17,189.97) .. (206.48,189.97) .. controls (201.8,189.97) and (198,186.17) .. (198,181.48) -- cycle ;
\draw [line width=1.5]  [dash pattern={on 5.63pt off 4.5pt}]  (122.77,106.63) .. controls (165.77,108.63) and (213.77,90.63) .. (250.77,9.63) ;
\draw [line width=1.5]  [dash pattern={on 5.63pt off 4.5pt}]  (223.77,66.63) .. controls (295.77,92.63) and (220.77,166.63) .. (293.77,202.63) ;

\end{tikzpicture}}
\end{minipage}
\end{tabular}
\caption{\sl
Illustration of a labeling induced by other labelings. First, consider the three graphs on the left. The top-left graph is $G_1$, which has a labeling $\ell_1$ represented by the node colored white, and $H_1$ is the subgraph that is above the dashed line.
Similarly, $G_2$ is the top-right graph, with labeling $\ell_2$ represented by color light grey, and $H_2$ is on the right of the dashed line. Finally, $G_3$ is the graph on the bottom, with labeling $\ell_3$ represented by color  dark gray, and $H_3$ is below the dashed line. Now, the graph on the right has a labeling that is induced by $(G_1,\ell_1)$, $(G_2,\ell_2)$ and $(G_3,\ell_3)$, via $H_1$, $H_2$ and $H_3$. }
\label{fig:induced-labeling}
\end{figure}

We also define a notion of \emph{boundary}.

\begin{definition}[Boundary]
Let $G$ be a graph, and $H$ be a subgraph of $G$. The \emph{boundary} of $H$ in~$G$, denoted by $\partial_G H$ is the set of nodes of $V(H)$ that are incident to an edge in $E(G)\setminus E(H)$.
\end{definition}

We are now ready to define local stability.

\begin{definition}
\label{def:local-stability}
A language $\cL$ is \emph{locally stable} if there exists a constant $\beta>0$, such that, for every labeled graph $(G,\ell)$ and for every $k$, the following holds. For every labeled graphs $(G_i,\ell_i)\in \cL$, $i=1,\dots,k$, and every subgraphs $H_1,\dots,H_k$, such that $(G,\ell)$ is induced by the labeled graphs $(G_i,\ell_i)$ $i=1\dots k$ via the graphs $H_i$, $i=1\dots k$, the following holds:
\[
\dist ((G,\ell),\cL) \leq \beta \; |\cup_{i=1}^k (\partial_{G}H_i \cup \partial_{G_i}H_i)|.
\]
\end{definition}

Intuitively the definition says that by taking pieces of labelings from different correct instances (that might not use the same underlying graph), we get a labeling whose distance to the language is at most the size of the boundary between the pieces, up to some multiplicative constant. 
Note that $\partial_{G}H_i \cup \partial_{G_i}H_i$ measures the size of the boundary of the $i$-th piece in the patchwork instance, but also in the correct instance it comes from.

Our characterization is the following.

\begin{theorem}\label{thm:characterization}
Let $\cL$ be a distributed language. $\cL$ admits an error-sensitive proof-labeling scheme if and only if $\cL$ is locally stable.
\end{theorem}

More precisely, we establish that a language with sensitivity $\alpha$ is locally stable with parameter $\beta=\frac{1}{\alpha}$, and that a language with local stability $\beta$ is error-sensitive with parameter $\alpha=\frac{1}{\beta+1}$. We do not know whether these relations are tight or not.

\begin{proof}
We first show that if  a distributed language $\cL$ admits an error-sensitive proof-labeling scheme then $\cL$ is locally stable. So, let $\cL$ be a distributed language, and let $(\prover,\verif)$ be an error-sensitive proof-labeling scheme for $\cL$ with sensitivity parameter $\alpha$. Let $(G,\ell)$ be a labeled graph induced by labeled graphs $(G_i,\ell_i)\in \cL$, $i=1,\dots,h$, via the subgraphs $H_1,\dots,H_h$ for some $h\geq 1$. Since, for every $i\in\{1,\dots,h\}$, $(G_i,\ell_i)\in \cL$, there exists a certificate function $c_i$ such that $\verif$ accepts at every node of $(G_i,\ell_i)$ provided with the certificate function $c_i$.
Now, let us consider the labeled graph $(G,\ell)$, with certificate $c_i(u)$ on every node $u\in V(H_i)$ for all $i=1,\dots,h$. With such certificates, the nodes in $V(H_i)$ that are not in $\partial_{G}H_i \cup \partial_{G_i}H_i$ have the same close neighborhood  in $(G,\ell)$ and in $(G_i,\ell_i)$. Therefore,  they accept in $(G,\ell)$ the same way they accept in $(G_i,\ell_i)$. It follows that the number of rejecting nodes is bounded by $|\cup_{i=1}^h (\partial_{G}H_i \cup \partial_{G_i}H_i)|$, and therefore $(G,\ell)$ is at Hamming distance at most $\frac{1}{\alpha} |\cup_{i=1}^h (\partial_{G}H_i \cup \partial_{G_i}H_i)|$ from~$\cL$. Hence, $\cL$ is locally stable, with parameter $\beta = \frac{1}{\alpha}$. 

It remains to show that if a distributed language is locally stable then it admits an error-sensitive proof-labeling scheme. Let $\cL$ be a locally stable distributed language with parameter $\beta$. We prove that the universal proof-labeling scheme $(\prover_{univ},\verif_{univ})$ for $\cL$ (cf. Section~\ref{sec:warmup}) is error-sensitive for some parameter $\alpha$ depending only on $\beta$. Let $(G,\ell)\notin\cL$, and let us fix some certificate function $c$. The verifier $\verif_{univ}$ rejects in at least one node. We show that if $\verif_{univ}$ rejects at $k$ nodes, then the Hamming distance between $(G,\ell)$ and $\cL$ is at most $k/\alpha$ for some constant $\alpha>0$ depending only on $\beta$. For this purpose, let us consider the outputs of $\verif_{univ}$ applied to $(G,\ell)$ with certificate $c$, and let us define the graph $G'$ as the graph obtained from $G$ by removing all edges for which $\verif_{univ}$ rejects at both extremities. Note that the graph $G'$ may not be connected. 

Let $C$ be a connected component of $G'$, with at least one node $u$ at which $\verif_{univ}$ accepts.
Recall that we used the notation $(T,M,L)$ for the certificates of the universal scheme (cf. Section~\ref{sec:warmup}). 
We claim that all the vertices of $C$ have received the same certificate $(T,M,L)$. 
Indeed, if it is not the case, then, by connectivity there exist two vertices that are adjacent in~$C$, and that do not have the same certificate. 
This is a contradiction. Indeed, these two vertices would have detected the inconsistency, and would have both rejected, thus the edge between them would have been removed. We denote by $(G_C,\ell_C)$ the labeled graph described by $(M,L)$. 
In addition, since  $\verif_{univ}$ accepts in at least one node $u$, it must be that $(G_C,\ell_C)\in\cL$.
Finally, we prove that $C$ is a subgraph of $G_C$, and that the labeling $\ell$ and $\ell_C$ coincide on $C$.
Consider an edge of $C$. Necessarily, at least one of its endpoints is accepting (otherwise this edge would have been removed). If the vertex accepts, it means that this edge exists in $G_C$, and that both endpoints have the same label in $\ell$ and $\ell_C$. 

Let us now consider the other possibility: $C$ is a connected component of $G'$ where all nodes reject. 
By construction, such a component is composed of just one isolated node. 
For every such isolated rejecting node~$u$, let us denote by $(G_C,\ell_C)$ a labeled graph composed of a unique node, with ID equal to the ID of $u$, and with labeling $\ell_C(u)$ such that $(G_C,\ell_C)\in \cL$.

Let $\cal C$ be the set of all connected components of $G'$.
Note that $\cal{C}$ is a partition of the vertices of $G$, and that we established that for every $C$, $C$ is a subgraph of $G_C$, 
and the labelings $\ell_C$ and $\ell$ coincides on $C$.
Therefore we can legally define $(G,\ell')$ as the graph induced by labeled graphs $(G_C,\ell_C)$ via the subgraphs $C\in \cal C$. 
By local stability, we get the following:
\[
\dist((G,\ell'),\cL) \leq  \beta \,|\cup_{C\in \cal C} (\partial_{G}C \cup \partial_{G_C}C)|.
\]
Now, let us consider the number $k$ of nodes rejecting $(G,\ell)$.
By construction, the nodes in $\cup_{C\in \cal C} (\partial_{G}C \cup \partial_{G_C}C)$ are exactly the nodes that are rejecting $(G,\ell)$, thus:
\[
k=|\cup_{C\in \cal C} (\partial_{G}C \cup \partial_{G_C}C)|.
\]
Finally, again by construction, the Hamming distance between $(G,\ell')$ and $(G,\ell)$ is at most the number of isolated rejecting nodes, which implies:
\[ 
\dist((G,\ell'),(G,\ell))\leq k.
\]
Putting all the pieces together we get:
\[
\dist((G,\ell),\cL) \leq (\beta+1) \, k.
\]
In other words, the universal proof-labeling scheme is error-sensitive, with parameter $\alpha=\frac{1}{\beta+1}$.
\end{proof}

Theorem \ref{prop:notallwork} can be viewed as a corollary of Theorem~\ref{thm:characterization} as it is easy to show that $\mbox{\sc regular}$ is not locally stable. Nevertheless, local stability may not always be as easy to establish, because it is based on merging an arbitrary large number of labeled graphs. 
We thus consider another property, called \emph{strong local stability}, which is easier to check, and which provides a sufficient condition for the existence of an error-sensitive proof-labeling scheme. 
Given two labeled graphs $(G,\ell)$ and $(G',\ell')$, and a subgraph $H$ of both $G$ and $G'$, we define a third labeling of $G$, that we call $\ell-\ell_H+\ell'_H$. 
For every node $v\in V(G)$:
\[
(\ell-\ell_H+\ell'_H)(v) = \left\{ \begin{array}{ll}
 \ell'_H(v) & \text{ if } v\in V(H), \\ 
 \ell(v) & \text{ otherwise.}
 \end{array}\right.
\]
To avoid double subscript, in the following we will sometimes use superscripts  instead of subscripts for sequences, e.g., $\ell^i$ instead of $\ell_i$.

\begin{definition}
A language $\cL$ is \emph{strongly locally stable} if there exists a constant $\beta>0$, such that, for every  graph $H$, and every two labeled graphs $(G,\ell)\in \cL$ and $(G',\ell')\in \cL$ admitting $H$ as a subgraph, the labeled graph $(G,\ell-\ell_H+\ell'_H)$ is at hamming distance at most $\beta \; |\partial_{G'}H+\partial_{G}H|$ from~$\cL$. 
\end{definition} 

The following theorem states that strong local stability is indeed a notion that is at least as strong as  local stability. 

\begin{theorem}
\label{thm:sufficientcond}
If a language $\cL$ is strongly locally stable, then it is locally stable. 
\end{theorem}

\begin{proof}
Let us consider a strongly locally stable language $\cL$, with parameter $\beta$,  and a labeled graph $(G,\ell)$ induced by labeled graphs $(G_i,\ell^{i})\in \cL$, $i=1,\dots,h$, via the subgraphs $H_1,\dots,H_h$.
We will establish that $\dist ((G,\ell),\cL) \leq \beta \; |\cup_{j=1}^h (\partial_{G}H_j \cup \partial_{G_j}H_j)|$, which is the condition of local stability.

For a labeling $\ell'$ of $G$, let $\dist_i((G,\ell),(G,\ell'))$ be the distance between the labelings $\ell$~and~$\ell'$, restricted to $\cup_{j=1}^{i} H_j$ 
(in other words the label differences in $\cup_{j=i+1}^{h} H_j$ do not count for $\dist_i$).
We consider two sequences  of labelings of $G$, $\rho^{i}$ for $i=0,\dots,h$, and $\mu^{i}$ for $i=1,\dots,h$. 
They are defined iteratively in the following way. 
We take $\rho^{0}$ to be an arbitrary labeling such that $(G,\rho^{0})\in \cL$. 
For $i\geq 1$,  $\rho^{i}$ is a labeling such that $(G,\rho^{i})\in \cL$, and:
\begin{equation}
\label{eq:hypothese-i}
\dist_i\left(
(G,\ell),
(G,\rho^{i})
\right)
\leq 
\beta \sum_{j=1}^{i} |\partial_G H_j \cup \partial_{G_j}H_j|.
\end{equation}
Finally, we set 
\[
\mu^{i}=\rho^{i-1} - \rho^{i-1}_{H_i} + \ell_{H_i},
\]
Note that this labeling satisfies the distance inequality, because $\dist_0$ is always zero.
To prove our result, it is sufficient to show that we can indeed define the sequence $\rho^{i}$, $i=0,\dots,h$. 
Indeed, if we get to $\rho^{h}$, then since $(G,\rho^{h})\in \cL$ and $\dist_h=\dist$, Equation~\ref{eq:hypothese-i} transforms into  $ \dist\left(
\left(G,\ell\right),
\cL
\right)
\leq 
\beta \sum_{j=1}^{h} |\partial_G H_j \cup \partial_{G_j}H_j|$, and because the sets 
$(\partial_G H_j \cup \partial_{G_j}H_j)_j$ 
are disjoint, the right-hand side is equal to 
$\beta  |\cup_{j=1}^{h}\partial_G H_j \cup \partial_{G_j}H_j|$, which is what we want.

By induction, suppose that we have built a proper $\rho^{i}$.
To define $\mu^{i+1}$, we take $\rho^i$ and copy the labeling of $\ell$ on a still untouched subgraph $H_{i+1}$. 
Therefore: 
\begin{equation}
\label{eq:mu-rho}
\dist_{i+1}\left(
(G,\ell),
(G,\mu^{i+1})
\right)
=
\dist_{i}\left(
(G,\ell),
(G,\rho^{i})
\right)
\leq 
\beta \sum_{j=1}^{i} |\partial_G H_j \cup \partial_{G_j}H_j|.
\end{equation}
We take $\rho^{i+1}$ to be a labeling such that $(G,\rho^{i+1})$ is in $\cL$, and the distance between $(G,\rho^{i+1})$ and  $(G,\mu^{i+1})$ is minimized.
By strong local stability, since $\ell$ and $\rho^{i}$ are legal labelings for $\cL$, we get that: 
\begin{equation}
\label{eq:strong}
\dist((G,\mu^{i+1}),(G,\rho^{i+1}))= \dist((G,\mu^{i+1}),\cL) \leq \beta 
|\partial_G H_{i+1} \cup \partial_{G_{i+1}}H_{i+1}|.
\end{equation}
Putting the Equations~\ref{eq:mu-rho} and~\ref{eq:strong} together, by triangle inequality, we get:
\[
\dist_{i+1}\left(
(G,\ell),
(G,\rho^{i+1})
\right)
\leq 
\beta \sum_{j=1}^{i+1} |\partial_G H_j \cup \partial_{G_j}H_j|.
\]
That is, we get Equation~\ref{eq:hypothese-i} at index $i+1$, which proves the theorem by induction. 
\end{proof}

In fact, strong local stability is  a notion strictly  stronger than local stability, although they coincide on bounded-degree graphs. 

\begin{theorem}
\label{prop:strictlyandbounded}
There are languages that are locally stable but not strongly locally stable. However, all locally stable languages on bounded degree graphs are strongly locally stable. 
\end{theorem}

\begin{proof}
Let us define a language $\cL$ to prove the first part of the theorem. 
As earlier in the paper (e.g., in the proof of Theorem~\ref{prop:notallwork}), a proper labeling $\ell$ for $\cL$ describes a set of edges $H_\ell$. Here, in addition, every node is also assigned a color: blue or red. 
The labeling is in the language $\cL$ if every connected component of $H_{\ell}$ is monochromatic.
 
This language has a proof-labeling scheme with empty certificates. 
The verifier simply checks that $H_{\ell}$ is well-defined, and that every neighbor in $H_\ell$ has been given the same color.
In addition, this scheme is error-sensitive. This is because, for every inconsistency in the description of $H_{\ell}$, or any edge of $H_{\ell}$ that is not monochromatic, both endpoints reject. As a consequence, if every rejecting node modifies its local description of $H_{\ell}$ by removing the faulty edges, the new labeling is in the language. In turn, this means that the distance from the language is upper bounded by the number of rejecting nodes.
By Theorem \ref{thm:characterization}, we know that the language $\cL$ is locally stable.  

We show that $\cL$ is not strongly locally stable. 
Consider a graph $G$ that is a star with $2p$ leaves. 
Now, consider two labelings $\ell$ and $\ell'$ where $H_{\ell}=H_{\ell'}=G$, and all the nodes are blue in $\ell$ and red in $\ell'$. 
Let $H$ be a subgraph of $G$ with the center and $p$ leaves.
We note that $(G,\ell-\ell_H+\ell'_H)$ is at distance $p$ from $\cL$. This is because the best we can do is to edit the labels of all the vertices of $G \setminus H$.
On the other hand, $\partial_{G}H$ contains only one node, the center.
As we can make $p$ arbitrarily large, the condition of strong local checkability cannot be fulfilled.

We now show that all locally stable languages on bounded degree graphs are strongly locally stable. Let $\Delta\geq 1$, and let ${\cal F}_\Delta$ be the family of graphs with maximum degree $\Delta$. Let $\cL$ be a locally stable language on graphs in ${\cal F}_\Delta$. 
Let us consider a connected graph $H$, and two labeled graphs $(G,\ell)\in \cL$ and $(G',\ell')\in \cL$, with $G\in {\cal F}_\Delta$, and $G'\in {\cal F}_\Delta$, both admitting $H$ as a subgraph. 
Let $(G,\ell-\ell_H+\ell'_H)$ be the labeled graph induced by $(G,\ell)$ and $(G',\ell')$ via the subgraph $H$. 
We view $(G,\ell-\ell_H+\ell'_H)$ as induced by $(G,\ell)$ and $(G',\ell')$ via the subgraphs $G\setminus H$ and $H$. 
By local stability, we get that the distance from $(G,\ell-\ell_H+\ell'_H)$ to $\cL$ is at most $\beta \; |\left(\partial_{G}H \cup \partial_{G'}H\right) \cup \left(\partial_{G}(G \setminus H)\right)|$. 
(Note that for $G\setminus H$, $G$ is both the original graph, and the one that induces the labeling, hence there is just one border to consider.) 
Now,  $|\partial_{G}(G \setminus H)|\leq \Delta |\partial_{G}H|$, because each edge from the cut $(H,G\setminus H)$ must have an endpoint in $H$ and these endpoints have  degree at most $\Delta$. As a consequence the distance from $(G,\ell-\ell_H+\ell'_H)$ to $\cL$ is at most $\beta(\Delta+1)|\partial_{G}H \cup \partial_{G'}H|$, and the strong local stability follows.
\end{proof}

We do not have examples of ``natural'' languages that are locally stable but not strongly locally stable. In fact, the rest of this section is devoted to using strong local checkability applied to various ``natural'' languages.
Let us give an example where strong local stability is useful for easily proving error-sensitivity. 
Consider the following language.

\begin{align*}
\mbox{\sc leader}= \big\{(G,\ell) : \; 
& \forall v\in V(G), \, \ell(v)\in\{0,1\}, \;\\
& \mbox{and there exists a unique $v\in V(G)$ for which $\ell(v)=1$} \big\}.
\end{align*}

\begin{corollary}\label{coro:leader}
$\mbox{\sc leader}$  admits an error-sensitive proof-labeling scheme. 
\end{corollary}

\begin{proof}
Consider an arbitrary graph $H$, and two labeled graphs $(G,\ell)$ and $(G',\ell')$ in $\mbox{\sc leader}$. 
On the one hand, in $(G,\ell-\ell_H+\ell'_H)$, there can be only 0, 1, or 2 vertices with $\ell(v)=1$. 
On the other hand, $|\partial_{G'}H+\partial_{G}H|$ is at least 1, by connectivity. 
Therefore we get that $(G,\ell-\ell_H+\ell'_H)$ is at Hamming distance at most $2 \; |\partial_{G'}H+\partial_{G}H|$ from the language, thus that language is strongly locally stable, and the corollary follows from Theorems~\ref{thm:characterization}
and~\ref{thm:sufficientcond}. 
\end{proof}

Also, one can show that the language $\ST_l$ of spanning trees, whenever encoded by adjacency lists, admits an error-sensitive proof-labeling scheme, in contrast to Theorem~\ref{prop:STnot}. 

\begin{corollary}
\label{coro:ST-via-characterization}
$\mbox{\sc st}_l$  admits an error-sensitive proof-labeling scheme. 
\end{corollary}

\begin{proof}
We show that $\mbox{\sc st}_l$ is strongly locally stable. Let us consider two labeled graphs $(G,\ell)\in \mbox{\sc st}_l$ and $(G',\ell')\in\mbox{\sc st}_l$, both admitting $H$ as a subgraph. We show that $(G,\ell-\ell_H+\ell'_H)$ is not far from $\cL$. For this purpose, we aim at modifying the labels of few nodes so that to form a spanning tree of $G$. First, for every node $u \in \partial_{G}H \cup \partial_{G'}H$, we modify $\ell'_H(u)$ such that the label of $u$ becomes consistent with its neighborhood in $G$. That is, all edges listed in the label exist in $G$, and they match edges listed by the neighbors of $u$ in $G$. After this modification, which impacts only $|\partial_{G}H \cup \partial_{G'}H|$ nodes, the resulting labeling of the nodes in $G$ encodes a set of edges $F\subseteq E(G)$. However, $F$ may not be a spanning tree, as it may include cycles, and may even be not connected.  

Let $\widehat{G}$ be the graph obtained from $G$ after removing all edges in $E(H)$, and all nodes in $V(H)\setminus (\partial_GH \cup \partial_{G'}H)$. Note that $V(H) \cup V(\widehat{G})=V(G)$ and $V(H) \cap V(\widehat{G})=\partial_GH \cup \partial_{G'}H$. 
The set $F$ is equal to the union of the edges described by $\ell$ on $\widehat{G}$, and of the edges described by $\ell'$ on~$H$. Indeed consider an edge $e\in F$. If both endpoints of $e$ are in $\widehat{G}$, then this edge is encoded by $\ell$ at its two endpoints, as the labels of these endpoints are copied from $\ell$, and the modification of $\ell-\ell_H+\ell'_H$ performed at the nodes in $\partial_{G}H \cup \partial_{G'}H$ does not impact such nodes. If $e$ has both endpoints in $H \setminus (\partial_GH \cup \partial_{G'}H)$ then, by the same reasoning,  this edge is encoded by $\ell'$ at its two endpoints. If $e$ has both endpoints in $\partial_GH \cup \partial_{G'}H$, then the modification of $\ell-\ell_H+\ell'_H$ performed at the nodes in this latter set did not affected edge $e$, which implies that $e$ was originally encoded in $\ell'$. Finally, if $e$ has one endpoint in $\partial_GH \cup \partial_{G'}H$, and the other one outside $\partial_GH \cup \partial_{G'}H$, then, from by the modification of $\ell-\ell_H+\ell'_H$, the edge $e$ was present in $\ell$ in at least one of its extermities. 
 
As $\ell$ is the labelling of a spanning tree of $G$, $F$ restricted to $\widehat{G}$ is a spanning forest of $\widehat{G}$. Similarly, as $\ell'$ is a spanning tree of $G'$, $F$ restricted to $H$ is a spanning forest of $H$. Also, since $V(\widehat{G})\cap V(H)=\partial_GH \cup \partial_{G'}H$, it follows that, in both forests, every tree contains a node of $V(\widehat{G})\cap V(H)$. Let us denote by $n_{\widehat{G}}$, $m_{\widehat{G}}$, and $s_{\widehat{G}}$ the number of nodes, edges, and connected components of $F$ restricted to $\widehat{G}$, respectively. Similarly, let us denote by $n_H$, $m_H$, and $s_H$ the same parameters for~$H$. Since the connected components of $F$ restricted to $\widehat{G}$, and to $H$, are forests, we get that:
\begin{equation}\label{eq1}
m_{\widehat{G}}=n_{\widehat{G}}-s_{\widehat{G}}, 
\text{ and }
m_H=n_H-s_H.
\end{equation}
Moreover, since each connected component contains a node of the border, we get 
\begin{equation}\label{eq2}
s_{\widehat{G}}\leq |V(\widehat{G})\cap V(H)|, \; \mbox{and} \; s_H\leq |V(\widehat{G})\cap V(H)|.
\end{equation}
Now, let us  consider the whole set $F$, and let us define $n_F$, $m_F$, and $s_F$ as the number of nodes, edges, and connected components of $F$, respectively.  By definition, $m_F=m_{\widehat{G}}+m_H$. Thus, by Eq.~\eqref{eq1}, we get that 
\[
m_F=n_{\widehat{G}}+s_{\widehat{G}}+n_H+s_H.
\]
Moreover, by definition, $n_F=n_{\widehat{G}}+n_H-|V(\widehat{G})\cap V(H)|$. Therefore, 
\[
m_F=n_F+|V(\widehat{G})\cap V(H)|+s_{\widehat{G}}+s_H.
\] 
We can now bound the number of edges that we need to remove from $F$ in order to get a spanning forest (with the same number of connected components). For such a forest, it must hold that its number of edges, $m$, satisfies $m=n_F+s_F$. Therefore, 
\begin{eqnarray*}
m_F -m & = & (n_F+|V(\widehat{G})\cap V(H)|+s_{\widehat{G}}+s_H)-(n_F+s_F) \\
& \leq &  |V(\widehat{G})\cap V(H)|+s_{\widehat{G}}+s_H \\
& \leq & 3|V(\widehat{G})\cap V(H)|,
\end{eqnarray*}
where the last equality holds by Eq.~\eqref{eq2}.
Thus, by removing at most $3|\partial_GH \cup \partial_{G'}H|$ edges from $F$, we get a spanning forest of $G$ with at most $|\partial_GH \cup \partial_{G'}H|$ connected  components. 
Therefore, by adding $|\partial_GH \cup \partial_{G'}H|-1$ edges, one can construct a spanning tree of $G$. 
So, in total, transforming $F$ into a spanning tree required to modify at most $4|\partial_GH \cup \partial_{G'}H|$ edges. 
This may impact the labels of at most $8|\partial_GH \cup \partial_{G'}H|$ nodes. 
As the labels of the nodes in  $\partial_GH \cup \partial_{G'}H$ were also modified at the very beginning of the construction, it follows that the number of node labels impacted by our spanning tree construction is at most $9|\partial_GH \cup \partial_{G'}H|$. 
It follows that $\mbox{\sc st}_l$ is strongly locally stable with parameter at most~9, which implies that it admit an error-sensitive proof-labeling scheme with sensitivity parameter at least $\frac{1}{9}$, by Theorem~\ref{thm:characterization}, and Theorem~\ref{thm:sufficientcond}.
\end{proof}

Also, Theorem~\ref{thm:characterization} allows us to prove that minimum-weight spanning tree (MST) is error-sensitive (whenever the tree is encoded locally by adjacency lists). More specifically, let 
\begin{equation}\label{eq:mstlist}
\mbox{\sc mst}_l = \Big\{(G,\ell) : \forall v\in V(G), \, \ell(v)\subseteq N(v) \; \mbox{and} \;\Big ( \bigcup_{v\in V(G)}\bigcup_{u\in \ell(v)} \{u,v\}\Big ) \; \mbox{forms a MST} \Big\}. 
\end{equation}

\begin{corollary}\label{cor:mstlocstable}
$\mbox{\sc mst}_l$  admits an error-sensitive proof-labeling scheme. 
\end{corollary}

\begin{proof}
We show that $\mbox{\sc mst}_l$ is strongly locally testable. 
Let us consider a graph $H$, and two labeled graphs $(G,\ell)\in \mbox{\sc mst}_l$ and $(G',\ell')\in \mbox{\sc mst}_l$ admitting $H$ as a subgraph. We show that the labeled graph $(G,\ell-\ell_H+\ell'_H)$ is not far from $\mbox{\sc mst}_l$. 
Let $T$ be the spanning tree of $G$ defined by the set 
of edges defined by $\ell$, and let $T'$ be the spanning tree of $G'$ defined by the set of edges defined by $\ell'$. Let $F$ the edge set defined by $\ell-\ell_H+\ell'_H$ on $G$, after the same modification of that labeling on the nodes of $\partial_{G}H \cup \partial_{G'}H$ as in the proof of Corollary \ref{coro:ST-via-characterization}, i.e., the labels of $\partial_{G}H \cup \partial_{G'}H$ are modified so that the adjacency lists of these nodes in their labels match the labels of their neighbors.
Let $\widehat{G}$ be the graph defined as in the proof of Corollary \ref{coro:ST-via-characterization}, that is, $\widehat{G}$ is the graph obtained from $G$ after removing all edges in $E(H)$, and all nodes in $V(H)\setminus (\partial_GH \cup \partial_{G'}H)$.
Note that $F$ is obtained from the union of the two forests that came from $\ell$ and $\ell'$, on  $E(\widehat{G})$ and $E(H)$, respectively. Hence,  every connected component of $F$ contains a node in $\partial_GH \cup \partial_{G'}H$. 

Recall that Kruskal algorithm constructs an MST by considering the edges in increasing order of their weights, and by adding the currently considered edge to the current set of edges if and only if this edge does not create a cycle with the previously added edges. 
It is known that every MST of a graph can be generated by Kruskal algorithm, by breaking ties between edges of identical weight in a way to  add all edges of the desired MST.
Let $\cO$ be the ordering of the edges of $G$ that leads to the tree $T$, and let $\cO'$ be the ordering of the edges of $G'$ that leads to the tree $T'$. 
Let $\cO'_H$, be the same ordering as $\cO'$ but restricted to the edges of~$H$. 

Let $G_1$ be the graph obtained from $H$ by adding a new node $u$ connected to every node of $\partial_GH+\partial_{G'}H$ by edges with weights smaller than the smallest weight in $E(G)$ and in $E(G')$. 
Let $\cO_1$ be the ordering of $E(G_1)$ obtained by concatenating $\cO_H'$ to an arbitrary ordering of the edges incident to $u$. Let $T_1$ be the MST of $G_1$ that Kruskal algorithm constructs in $G_1$ when it uses the ordering $\cO_1$. 
Also let $G_2$ be a copy of $H$, let $T_2$ be the MST constructed by Kruskal algorithm on $G_2$ using $\cO_2=\cO_H'$. 
Finally, we define the ordering $\cO_3$ of the edges of $G$ as the ordering such that the edges of $E(\widehat{G})$ appear in the same order as in $\cO$, the edges of $E(H)$ appear in the same order as in $\cO'$, and the edges of $E(T) \cap E(\widehat{G})$ have priority. Let $T_3$ be the spanning tree defined by Kruskal algorithm on $G$ with $\cO_3$. 
$T_3$ is necessarily equal to $T$ on the edges of $\widehat{G}$ because they are MST of the same graph, and because the edges of $E(T) \cap E(\widehat{G})$ have priority in $\cO_3$. We claim the following: 

\begin{claim}\label{claim:MST-inclusions}
The following inclusions hold.
\[
E(T_1) \cap E(H) \subseteq E(T') \cap E(H) \subseteq E(T_2) \cap E(H).
\]
\[
E(T_1) \cap E(H) \subseteq E(T_3) \cap E(H) \subseteq E(T_2) \cap E(H).
\]
\end{claim}

Before proving Claim \ref{claim:MST-inclusions}, let us show how to complete the proof using that Claim. By Claim~\ref{claim:MST-inclusions}, on $H$, $T_3$ can be transformed into $T'$ by changing only edges of $E(T_2) \setminus E(T_1)$. Moreover $E(T_2) \cap E(H)$ and $E(T_1) \cap E(H)$ are a spanning forests of $H$ with at most $|\partial_{G}H \cup \partial_{G'}H|$ trees in it, because, as in the proof of Corollary \ref{coro:ST-via-characterization}, every tree contains at least a node of $\partial_{G}H \cup \partial_{G'}H$. We get that 
\[
|(E(T_2) \cap E(H)) \setminus (E(T_1) \cap E(H))|\leq |\partial_{G}H \cup \partial_{G'}H|.
\] 
Therefore, restricted to the graph $H$, the tree $T_3$ can be transformed into the tree $T'$ by adding or removing at most $|\partial_{G}H \cup \partial_{G'}H|$ edges. 
Now, as $T_3$ is equal to $T$ on $\widehat{G}$, $E(T_3)$ can be transformed into $F$ by changing at most $|\partial_{G}H \cup \partial_{G'}H|$ edges. Thus $F$ is at Hamming distance at most $2|\partial_{G}H \cup \partial_{G'}H|$ from a MST of $G$. Since the modification we made at the very beginning to ensure the consistency of the labels affected at most $|\partial_{G}H \cup \partial_{G'}H|$ nodes, it follows that the Hamming distance from $(G,\ell-\ell_H+\ell'_H)$ to the language is most $3|\partial_{G}H \cup \partial_{G'}H|$, and thus the language is strongly locally stable. This completes the proof of Corollary~\ref{cor:mstlocstable}, assuming Claim \ref{claim:MST-inclusions}. 

\medbreak

It just remains to prove Claim \ref{claim:MST-inclusions}. We show the two sets of inclusion at once. Let  $M$ be either $E(T')$ or $E(T_3)$, and let $\Omega$ be the ordering of the edges which makes Kruskal algorithm build $T'$ or $T_3$. 
Note that, by construction, $\Omega$, $\cO_1$, and $\cO_2$ are consistent on the edges that they have in common, i.e., on all the edges of $E(H)$. Let $\cO_{\text{tot}}$ be an ordering that is consistent with the three orderings $\Omega$, $\cO_1$ and $\cO_2$. 
We can run Kruskal algorithm on the three instances $G$, $G_1$ and $G_2$ with $\cO_{\text{tot}}$. Let $i\geq 1$, and let $M_1^{(i)}$, $M_2^{(i)}$ and $M^{(i)}$, be the subset of edges in $E(T_1)$, $E(T_2)$, and $M$, respectively, that have been added  to the current tree by Kruskal algorithm before considering the $i$th edge in $\cO_{\text{tot}}$. We show, by induction on $i$, that the three following properties hold for every $i\geq 1$:

\begin{description}
\item[P1:] $M_1^{(i)} \cap E(H) \subseteq M^{(i)} \cap E(H) \subseteq M_2^{(i)} \cap E(H)$; 
\item[P2:] if two nodes of $H$ are linked by a path in $M_2^{(i)}$ then they are  linked by a path in~$M^{(i)}$;
\item[P3:] if two nodes of $H$ are linked by a path in $M^{(i)}$ then they are  linked by a path in~$M_1^{(i)}$.
\end{description}

These properties are trivially true for $i=1$, as all sets $M_1^{(1)}$, $M_2^{(1)}$ and $M^{(1)}$ are empty. Suppose that P1, P2, and P3 hold are true for $i-1$, and consider $i$-th edge $e=\{u,v\}$ considered by Kruskal algorithm in $\cO_{\text{tot}}$ for $T_1$, $T_2$ and $T'$ or $T_3$. We consider two cases. 

Consider first the case where $e \notin E(H)$. Then $e$ appears either only in $\cO_1$, or only in $\Omega$.
If $e$ appears only in $\cO_1$, then independently of whether Kruskal algorithm takes $e$ or not, the three properties P1, P2, and P3 hold for $i$. 
If $e$ appears only in $\Omega$, then, clearly, P1 and P2 hold for $i$. 
The only scenario for which P3 may not hold for $i$ is if $e$ is added to $M$, and this addition creates a new path between two nodes $x$ and $y$ of $H$, while there are no paths between $x$ and $y$ in $M_1^{(i)}$. Let us show that this does not happen. Indeed, since $e\notin E(H)$, such a path must pass through the border of $H$, which is included in $\partial_GH\cup\partial_G'H$ (this holds for both choices for $M$, that is, either $E(T')$ or $E(T_3)$). 
In particular adding $e$ to the set of edges taken by Kruskal algorithm so far connects two nodes of the border of $H$. Now, all the nodes of $\partial_GH\cup\partial_G'H$ are already connected in $M_1^{(i)}$. Indeed,  the edges of $E(G_1)\setminus E(H)$ have smaller weights. Therefore, all the nodes of $\partial_GH\cup\partial_G'H$ are connected in $M_1^{(i)}$, and thus it is not possible that there is a path created by adding $e$ in $M$ that does not already exists in $M_1^{(i)}$.

Second, consider the case $e\in E(H)$. Then $e$ appears in all the orderings. Let us consider two subcases depending on whether or not $e$ is taken in $M$.
\begin{itemize}
\item If $e$ is taken in $M$, then $e$ is not closing a cycle in $M^{(i-1)}$, and thus, thanks to P2, $e$ is not closing any cycle in $M_2^{(i-1)}$ either. Thus $e$ is also taken in $T_2$, and P1 holds. P2 still holds as well since $e$ is added to both sets. If $e$ is taken in $T_1$ then P3 holds. Instead if $e$ is not taken in $T_1$, then its two extremities were already linked by a path, and P3 also holds.

\item If $e$ is not taken in $M$, then $e$ closes a cycle in $M^{(i-1)}$. Therefore, by P3, $e$ also closes a cycle in $M_1^{(i-1)}$, and thus it is not taken in $T_1$ either, and P1 holds. P3 still holds as we have added no edges to $M$. If $e$ is not taken in $T_2$ then P2 holds. And if $e$ is  taken in $T_2$, then the fact that $e$ is not taken in $M$ implies that the nodes were already connected, and thus again P2 holds.
\end{itemize}
This completes the proof of Claim \ref{claim:MST-inclusions}, and thus the proof of Corollary~\ref{cor:mstlocstable}. 
\end{proof}

\section{Compact error-sensitive proof-labeling schemes}
\label{sec:compact}

The characterization of Theorem~\ref{thm:characterization} together with Lemma~\ref{lem:univ} implies an upper bound of $O(n^2)$ bits on the certificate size for the design of error-sensitive proof-labeling schemes for locally stable distributed languages. In this section, we show that the certificate size can be drastically reduced in certain cases. 
As said in the introduction, we focus on the spanning tree and minimum spanning tree problems, as they play a central role in the theory of proof-labeling schemes, and in distributed computing in general.
It is known that these languages have proof-labeling schemes using respectively certificates of  $\Theta(\log n)$ bits~\cite{AwerbuchPV91,KormanKP10}, and $\Theta(\log^2 n)$ bits~\cite{KormanK07}. We show that these schemes are actually error-sensitive.  

Recall that Theorem~\ref{prop:STnot} proved that the language $\ST_p$ of spanning trees encoded at each node by a pointer to its parent does not admit any error-sensitive. Hence, we are interested in $\ST_{l}$, i.e., the language of spanning trees encoded by adjacency lists.
 
\begin{theorem}\label{theo:STlogn} 
$\ST_{l}$ has an error-sensitive proof-labeling scheme with certificates of size $O(\log n)$ bits, and sensitivity $\frac{1}{4}$.
\end{theorem}

\begin{proof}
We show that the classical proof-labeling scheme $(\prover,\verif)$ for $\ST_{l}$ is error-sensitive. Let us first remind precisely what this scheme is.

On instances of the language, i.e., on labeled graphs $(G,\ell)$ where $\ell$ encodes a spanning tree $T$ of $G$, the prover $\prover$ does the following. 
It chooses an arbitrary root $r$ of $T$, and then assigns to every node~$u$ a certificate $(I(u),P(u),d(u))$ where $I(u)=\ID(r)$, $P(u)$ is the ID of the parent of $u$ in the tree (where we consider that the root is its own parent), and $d(u)$ is the hop-distance in the tree from $u$ to $r$.
The verifier $\verif$ at every node $u$ first checks that:
\begin{itemize} 
\item the adjacency lists are consistent, that is, if $u$ is in the list of $v$, then $v$ is in the list of $u$;
\item there exists a neighbor of $u$ with ID $P(u)$, we denote it $p(u)$;
\item the node $u$ has the same root-ID $I(u)$ as all its neighbors in $G$;
\item $d(u)\geq 0$. 
\end{itemize} 

\noindent Then, the verifier checks that:
\begin{itemize}
\item if $\ID(u)\neq I(u)$ then $d(p(u)) = d(u)-1$, and for every other neighbor $w$ listed in $\ell$, $d(w)=d(u)+1$ and $p(w)=u$; 
\item if $\ID(u)= I(u)$ then $P(u)=ID(u)$, $d(u)=0$, and every neighbor $w$ of $u$ listed in $\ell$ satisfies $d(w)=d(u)+1$ and $p(w)=u$. 
\end{itemize}

The scheme and later steps of the proof are illustrated in Figure~\ref{fig:ST-total}.
 
\begin{figure}[tb]
\begin{center}
\begin{tabular}{cc}
\begin{minipage}{0.5\textwidth}
\includegraphics[scale=.9]{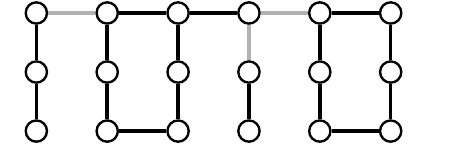}
\vspace{0.3cm}
\end{minipage}
&
\begin{minipage}{0.5\textwidth}
\includegraphics[scale=.9]{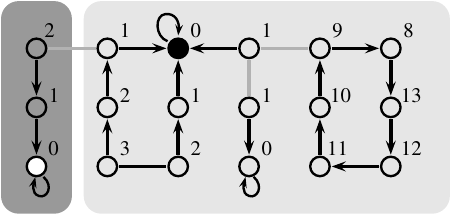}
\vspace{0.3cm}
\end{minipage}
\\
(a) Instance & (b) Certification 
\\
\\
\begin{minipage}{0.5\textwidth}
\includegraphics[scale=.9]{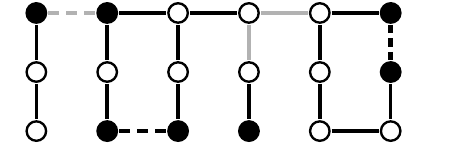}
\vspace{0.3cm}
\end{minipage}
&
\begin{minipage}{0.5\textwidth}
\includegraphics[scale=.9]{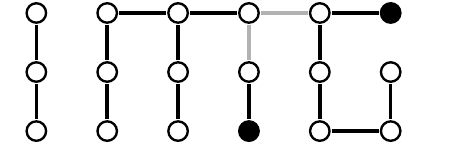}
\vspace{0.3cm}
\end{minipage}
\\
(c) Verifier output & (d)Graph $G'$
\end{tabular}
\end{center}
\caption{
\sl 
Construction in the proof of Theorem~\ref{theo:STlogn}.
(a) Illustration of an instance. The black edges are the edges described in $\ell$ (in this example, $\ell$ is a proper encoding of a set of edges), and
the gray edges are the other edges of the graph. Note that $(G,\ell)$ is not in the languages as the set of black edges in neither connected, nor~acyclic.
(b)~Illustration of a certificate assignment.
The depicted numbers are the hop-distances, and the arrows provide the orientation of the pointers. The nodes in the light grey area have the ID of the black node as their root-ID in their certificates.
The ones in the dark grey rectangle have the ID of the white node as their root-ID in their certificates.
(c) Illustration of the rejecting nodes and edges. The black nodes are the ones that reject, and the dashed edges have both endpoints rejecting. 
The two top left black nodes reject because they do not have the same root-ID. 
The ones linked by a dashed  edge at the bottom reject because they detect that the edge connecting them is in the input but is not oriented (that is, the condition ``$p(w)=u$'' does not hold). 
The third black node on the bottom rejects because it has a distance-counter equal to zero, and a pointer to itself, as if it was the root, but its ID is not the root-ID in its certificate. 
Finally, the endpoints of the dashed edge on the right reject because their distance counter are not consistent with the pointers. 
(d) The labeled graph $(G',\ell')$ obtained after removing the edges whose both endpoints reject. The black nodes are the ones that reject when we run again the verifier on this labeled  graph, with the same certificate.
}
\label{fig:ST-total}
\end{figure}

By construction, if $(G,\ell)\in \ST_{l}$, then $\verif$ accepts at every node. Conversely, if $(G,\ell)\notin \ST_{l}$, then, for every certificate function $c$,  at least one node rejects. 
To establish error-sensitivity for the above proof-labeling scheme, let us assume that $\verif$ rejects at $k\geq 1$ nodes with some certificate function $c$. 
Then, let $(G',\ell')$ be the labeled graph coinciding with $(G,\ell)$ except that all edges for which $\verif$ rejects at both endpoints are removed both from $G$, and from the adjacency lists in $\ell$ of the endpoints of these edges. 
Note that modifying $\ell$ into $\ell'$ only requires to edit labels of nodes that are rejecting.   
Note that the graph $G'$ might be disconnected. 
Also note that the edges described in the labeling that are still present in $(G',\ell')$ form a forest. 
(This is because the verification of the counters ensures that any cycle is cut by our procedure). 

 Let $(C,\ell'_C)$ be a connected component of $(G',\ell')$. 
We claim that the edges of $\ell'_C$ form a forest in $C$. 
 If there is a cycle in the edges of $\ell'_C$, then this cycle already existed in $\ell$ because  no edges were added when transforming $\ell$ into $\ell'$. Let us consider such a cycle in $\ell$ (if it exists), and let us consider the certificates given by $\prover$ to the nodes of this cycle. Either an edge is not oriented, i.e., no nodes use this edge to point to its parent, or the cycle is consistently oriented but some distances are not consistent. In both cases, two adjacent nodes of the cycle would reject when running $\verif$. It follows that this cycle cannot be present in $\ell'_C$, as at least one edge has been removed. As a consequence $\ell'_C$ form a forest of $C$. If a node is connected to no other nodes by an edge of $\ell'_C$, then we will consider such isolated node as a tree of a unique node. With this convention, $\ell'$ is a spanning forest of $G'$. 

We  now bound the number of trees in $\ell'$ by a function of $k$. The number of trees in $\ell'$ is equal to the sum of the number of trees in each component $(C,\ell'_C)$.  Let us run $\verif$ on the graph $(C,\ell'_C)$, and let $k_C$ be the number of rejecting nodes. 
Observe that, for every two nodes $u$ and $v$ in a component $C$, it holds that $I(u)=I(v)$. Indeed, otherwise, there would exist two adjacent nodes $u$ and $v$ in $C$ with $I(u)\neq I(v)$, resulting in $\verif$ rejecting at both nodes, which would yield the removal of $\{u,v\}$ from $G$. Consequently, at most one tree of $\ell'_C$ has a root whose ID corresponds to the ID given to the nodes in the certificate. 
Therefore, in every tree described in $\ell'_C$, except at most one, there exists at least one node that rejects (typically a vertex that is pointing to itself, but whose root-ID is different from its ID).
As a consequence, the number of trees in $\ell'_C$ is upper bounded by $k_C+1$, and the total number of trees is bounded by $\sum_C (k_C+1)=\left(\sum_C k_C\right) + |C|$.
Now, by construction of the proof-labeling scheme, the nodes that accept when running $\verif$ on $(G,\ell)$ also accept in $(G',\ell')$. Therefore $\sum_C k_C \leq k$. 

For every connected component $C$, let $V_C$ be the set of nodes of $C$. We claim that there exists a node of $V_C$ that rejects when we run $\verif$ on $(G,\ell)$. 
Suppose by contradiction that there exists a component where no node reject. Then no edges between $C$ and the rest of the graph would have been removed, and therefore there would be only one component in the graph. And then, as we know that at least one node rejects there would be a contradiction.  Therefore $|C|\leq k$. 
It follows that $\ell'$ encodes a spanning forest with at most $2k$ trees in total. Such a labeling can thus be modified to get a spanning tree by modifying the labels of at most $4k$ nodes. That is, $(\prover,\verif)$ is error-sensitive with parameter $\alpha\geq \frac14$. 
\end{proof}

Last, but not least, we show that the compact proof-labelling scheme in \cite{KormanK07,KormanKP10} for minimum-weight spanning tree (MST), as specified in Eq.~\eqref{eq:mstlist} of Section~\ref{sec:characterization}, is error-sensitive whenever the edge-weights are pairwise distinct.

\begin{theorem}\label{theo:MSTcompactES}
$\MST_l$ admits an error-sensitive proof-labeling scheme with certificates of size $O(\log^2n)$ bits, and sensitivity $\frac{1}{7}$.
\end{theorem}

The proof of Theorem \ref{theo:MSTcompactES} is quite technical. So, before entering into the details of the proof, let us provide the reader with some intuitions. A classic proof-labeling scheme for \MST\/ (see, e.g., \cite{KorKP11,KormanK07,KormanKP10}) consists in encoding a run of Bor\r{u}vka algorithm. Recall that Bor\r{u}vka algorithm maintains a spanning forest whose trees are also called fragments. Starting with the forest in which every node forms a fragment, Bor\r{u}vka algorithm proceeds in a sequence of steps where, at each step, the lightest  edge outgoing from every fragment of the current forest is selected, and is added to the \MST. The fragments linked by the selected edges are merged, and the algorithm goes to the next step. This algorithm eventually produces a single fragment, which is a \MST\/ of the whole graph, after at most a logarithmic number of steps. 

At each node $u$, the certificate of the classical proof-labeling scheme for \MST\/ consists of a table with a logarithmic number of fields, one for each step of Bor\r{u}vka algorithm. The corresponding entry of the table provides a proof of correctness for the fragment  including~$u$ at this step, plus the certificate of a tree pointing to the lightest outgoing edge of the fragment. The verifier verifies the structures of the fragments, and the fact that no  edges outgoing from each fragment have smaller weights than the one given in the certificate. It also checks that the different fields of the certificate are consistent. In particular, it checks that, if two adjacent nodes are in the same fragment at the same step, then they are also in the same fragment at the next step. 

To prove that this classic scheme is actually error-sensitive, we perform the same decomposition as in the proof of Theorem~\ref{theo:STlogn}, by removing the edges that have both endpoints rejecting. We then consider each connected component~$C$ of the remaining graph, and the subgraph $S$ of that component induced by the edges of the given labeling. In general, $S$ is not a \MST\/ of the component $C$, as it can even be disconnected. Nevertheless, we can still make use of a key property, which is that the subgraph~$S$ is not far from a \MST\/ of~$C$. Indeed, the edges of $S$ form a forest, and these edges belong to a \MST\/ of the component. As a consequence, it is sufficient to add a few edges to $S$ for obtaining a \MST. Thus, to show that $S$ is indeed not far from being a \MST\/ of $C$, we define a relaxed version of Bor\r{u}vka algorithm, and show that the labeling of the nodes corresponds to a proper run of this modified version of  Bor\r{u}vka algorithm. We then show how to slightly modify both the run of the modified  Bor\r{u}vka algorithm, and the labeling of the nodes, to get a \MST\/ of the component.  Finally, we prove that the collection of \MST{s} of the components can be transformed into a single \MST\/ of the whole graph, by editing  a few node labels only. 

The rest of the section is dedicated to formalizing the above intuition. 

\medbreak

\noindent\textbf{Proof of Theorem~\ref{theo:MSTcompactES}. \;} 
We show that the proof-labeling scheme for MST described in~\cite{KormanK07,KormanKP10} is error-sensitive. Let $G$ be an edge-weighted graph. For simplicity we assume that all the edge-weights of $G$ are distinct, and thus the MST is unique.  It is now folklore (see, e.g.,~\cite{Peleg00}) that  one can run a parallel version of Bor\r{u}vka algorithm which proceeds in at most $\lceil \log_2 n \rceil $ rounds, where each round consists in merging fragments in parallel. Note that a merging may involve more than just two fragments during a single round, so the number of fragments may actually decrease faster than by a factor~2 at each round. 

\medbreak

\noindent\emph{The standard proof-labeling scheme for \MST.} 
Recall that, in the proof-labeling scheme of~\cite{KormanK07,KormanKP10}, the prover $\prover$ essentially encodes at each node the run of the parallel version of Bor\r{u}vka algorithm. More specifically, the certificate at each node $u$ is divided into $\lceil\log_2 n\rceil+1$ fields, one for each round $i=1,\dots,\lceil \log_2 n \rceil $, plus an additional one. The field corresponding to round~$i$ in the certificate of a node $u$ contains: 
\begin{enumerate}
\item a pointer to a parent of $u$ in a rooted tree $T_1$ that is supposed to span the fragment including $u$ at round~$i$, rooted at an arbitrary node of the fragment, whose ID is the ID of the fragment, and the local proof used at $u$ to certify that $T_1$ is indeed a spanning tree of the fragment;  
\item a pointer to a parent of $u$ in another rooted tree $T_2$, also spanning the fragment, but rooted at the endpoint of the lightest edge $e$ outgoing the fragment, with the local proof at $u$ certifying $T_2$; 
\item the ID of the other endpoint of the edge $e$, and its weight. 
\end{enumerate}
The tree $T_1$  ensures the connectivity of the fragment, and is used in the merging procedure of Bor\r{u}vka algorithm, while $T_2$ is used to make sure that the edge $e$ is indeed the edge of minimum weight incident to the fragment. The additional field is used to encode and certify locally  the \MST\/ of the whole network.

The verifier $\verif$ checks that, for each round, the two spanning trees $T_1$ and $T_2$ of the fragment are correct. It also checks that the run is consistent, that is:
\begin{itemize}
\item two adjacent nodes with same fragment ID at some round have the same fragment ID and the same lightest outgoing edges for all further rounds; 
\item  if an edge is used to merge two fragments at some round, then its endpoints belong to the same fragment for all remaining  rounds;  
\item  if a spanning tree is pointing to an edge, then this edge exists, and it is used to merge the fragment with another fragment; 
\item the final spanning tree has exactly the edges described by the given labeling, and it correctly spans the whole graph,  i.e.,  all the nodes have the same root-ID for this tree. 
\end{itemize} 

It is proved in~\cite{KormanK07,KormanKP10} that $(\prover,\verif)$ is a proof-labeling scheme for $\MST$. We show that it is error-sensitive. 

\medbreak

\noindent\emph{Edge deletion.} Let us fix some some certificate function $c$, and let us assume that $k\geq 1$ nodes reject with certificate~$c$. We perform the same decomposition as in the proof of Theorem~\ref{theo:STlogn}, removing from $G$ and $\ell$ the edges whose two extremities are rejecting. 
We obtain a labeled graph $(G',\ell')$. Let $C$ be a connected component of $G'$, and let us run the verifier on $(C,\ell'_C)$ with the same certificate function. Let $k_C$ be the number of rejecting nodes in $C$.
As argued in the proof of Theorem~\ref{theo:STlogn}, the number of rejecting nodes in the whole graph can only decrease from $(G,\ell)$ to $(G',\ell')$. Therefore, $\sum_C k_C \leq k$. 

Let us consider a node that is rejecting in $(C,\ell'_C)$. We claim that the only cases for which a node rejects are (1)~it is not a root in one of the trees encoded in the certificates, but there are inconsistencies with its parent (i.e., no parent, or incorrect root-ID, or incorrect distance counter), or (2)~it is the root in one of the trees encoded in the certificates, but it is not incident to the edge announced in the certificate. This is because, using the same line of arguments as the proof of Theorem~\ref{theo:STlogn}, if another case of rejection would exist, then there would be an edge whose both endpoints reject, but such an edge cannot exist in $(C,\ell_C)$, by construction (these edges have precisely be removed when doing the decomposition).

\medbreak

\noindent\emph{Lazy Bor\r{u}vka.} We will now consider a relaxed version of Bor\r{u}vka algorithm that we call ``lazy Bor\r{u}vka''. As the classical version of Bor\r{u}vka, the lazy variant grows a forest of fragments. Initially, there is one fragment per node. At each round, lazy Bor\r{u}vka proceeds in three steps. First it picks an arbitrary \emph{name} for each fragment. Second, for each fragment~$F$, it considers all edges connecting $F$ to a fragment with different name, and either chooses the incident edge with smallest weight, or do not choose any edge, in which case we say that $F$ is \emph{skipping its turn}. Third, Lazy Bor\r{u}vka merges the fragments that are linked by edges selected during the second step.  The algorithm stops whenever all the fragments have the same name. 

Note that in general lazy Bor\r{u}vka does not produce an MST, and it may even not terminate. However, if the names assigned to adjacent fragments are distinct at each round, and if there is no round~$i$ such that all fragments skip their turn at every round $j\geq i$, then lazy Bor\r{u}vka eventually produces an MST. 

Given a fragment $F$, we refer to all fragments including $F$ during the further rounds of lazy Bor\r{u}vka as its \emph{successors}. The fragments of the previous rounds contained in $F$ called \emph{predecessors}. Also, a maximal set of adjacent fragments having the same name during a same round is called a \emph{cluster}. 

We restrict our attention to the runs of lazy Bor\r{u}vka satisfying the following two properties:

\begin{description}
\item[P1] 
If some fragments form a cluster, then all their successors will also be part of a same cluster, but they will remain in different fragments. 

\item[P2] 
At every round of the run, at most one fragment per cluster chooses an edge, and all the other skip their turn.
\end{description}

\noindent We  show that $\ell_C'$ corresponds to the outcome of a run of lazy Bor\r{u}vka satisfying the above two properties.

\begin{claim}\label{claim:correct-run}
The labeling $\ell_C'$ is the outcome of a correct run $R$ of lazy Bor\r{u}vka on $C$, and this run satisfies the properties P1 and P2.  
\end{claim}

\noindent \textit{Proof of the claim.} 
Let us consider an execution of  lazy Bor\r{u}vka where fragments are as described by the certificates, and the names of the fragments are the root-ID of the corresponding fragments. 
As we argued before, these fragments are well-defined, that is, they are trees with correct proof, and the same root-ID at every node. 
Moreover these fragments are consistent from any round to the next one, because they satisfy the consistency properties checked by the verifier.
The fact that the root-ID may not be the ID of the root of the tree is not a problem, as it corresponds to a name. Finally, recall that if a node of $C$ rejects when checking round $i$, this is because that node has no parent in a tree encoded in the certificate, and either it does not have the appropriate root ID, or it is not incident to the appropriate edge. In both cases, there are no outgoing edges corresponding to that fragment for round~$i$. We interpret this fact as the fact that this fragment skipped its turn at this round. 
It follows that the different steps are valid for lazy Bor\r{u}vka, and they correspond to~$\ell_C'$. 

We now prove that the run has property P1 and P2. 

For establishing P1, let us assume, for the purpose of contradiction, that, at some round in $R$, two adjacent fragments $F_1$ and $F_2$ have the same name, but two successors $F'_1$ and $F'_2$ have different names. If this is the case, then, when verifying the certificates, both endpoints of an edge $e$ connecting $F_1$ to $F_2$ reject. Indeed, the certificates describe this run, and the verifier checks that the rounds are consistent. There are no such edge $e$ in $C$ by construction of $G'$, thus this situation does not occur. Also, if the two successors $F'_1$ and $F'_2$ are identical, then, at some round, the certificates tell that a fragment is taking an edge to a fragment that has the same root-ID, which is impossible (as such an edge would have been removed when creating $G'$). These arguments generalize to clusters, by connectivity. 

For establishing P2, let us assume that, at some round $i$, two fragments of a cluster choose an edge. It means that in the certificates of this run, there were two fragments with correct spanning trees pointing to these edges. 
As the verifier checks that two adjacent nodes with the same root-ID have the same outgoing edge, this implies that either the outgoing edge~$e$ was the same in the certificates of the two fragments, or this edge was different for the two fragments. In the former case, at least one of the two fragments will take no edge because it will detect that it does not have the edge~$e$. In the latter case, all the edges between these fragments would have both endpoints rejecting, and then they would not be adjacent as these edges would have been removed. Again, these arguments generalize to cluster, by connectivity. 

Finally, the termination of the run is also correct with respect to the specification of lazy Bor\r{u}vka. 
This is because of two facts. First, in the certificate, the last field describes a tree that has the same root-ID at every node, and the verifier checks this. Thus this holds after the decomposition step. The run stops at a round where all the fragments have the same name. 
Second, suppose there was a round~$i$ before the last round described by the certificate at which all the fragments had the same name. 
Then thanks to P1, at round~$i$, the fragments were exactly the same as in the last round, and every node has skipped all the remaining rounds. Thus we can consider round~$i$ if it was the last round. It still holds that the run corresponds to $\ell_C$, and has property P1 and P2. This completes the proof of Claim~\ref{claim:correct-run}.  \hfill $\diamond$

\medskip 

\noindent\emph{Getting closer to an MST.} 
In general, whenever it terminates, lazy Bor\r{u}vka can produce a forest which is arbitrarily far from being an MST. Nevertheless, we show that, as the run $R$ satisfies the properties P1 and P2, the forest produced by this run is at distance at most $O(k_C)$ from an MST of $C$, where $k_C$ is the number of rejecting node in $C$. To do so, we now modify the run $R$, by applying iteratively an operation on the run, adding edges to $\ell_C'$. This addition of edges is repeated until one obtains a run where, at every round, not two adjacent fragments have the same name, that is, until one obtains a run that builds an MST. 

We now describe the operation that we apply to a run, and the labeling associated with the resulting run. Consider the first round for which there is a cluster with more than one fragment. Let $K$ be such a cluster. There are only a few cases to consider.

\begin{itemize}
\item Case 1: none of the fragments in $K$ is choosing an edge, although there are fragments with names different from the one of $K$ that are adjacent to $K$. In this case, for this round, we assign new distinct names to all of the fragments in $K$, making sure that these names are not already used at that round by other fragments --- that is, we use ``fresh'' names. 

\item Case 2:  one of the fragments of $K$ is choosing an edge that has minimum weight among all edges that connect this fragment to the other fragments of~$C$, including the fragments of~$K$. In this case, we replace the names of the other fragments of $K$ by fresh names. 

\item Case 3:  a fragment $F$ of the cluster $K$ is choosing an edge $e$, although the lightest edge outgoing from $F$ is an edge $e'$ that connects it to a fragment $F'$ of $K$. In this case,  we add a round between round $i-1$ and round~$i$ where all fragments of~$C$ are given distinct names, and every fragment is skipping its turn, except $F$, which is choosing the edge $e'$. Also we add this edge $e'$ to the labeling. 

\item Case 4: round $i$ is the last round. In this case, we have only one cluster $K$ containing all the remaining fragments. We consider the lightest edge $e$ connecting two fragments in $K$, and add a round between round $i-1$ and round~$i$ where all fragments of~$C$ are given distinct names, and every fragment is skipping its turn except~$F$, which is choosing the edge~$e$. Also we add this edge~$e$ to the labeling. 
\end{itemize}

\begin{claim}
\label{claim:P1P2preservedafteroperation}
If a correct run of lazy Bor\r{u}vka satisfies property P1 and P2, and corresponds to the current labeling, then the above operations preserve these properties.
\end{claim}

\noindent\textit{Proof of the claim.} 
Let us consider a correct run of lazy Bor\r{u}vka satisfyng property P1 and P2, and that corresponds to the current labeling. We consider the four cases of the operation, and establish the claim for each of them. 
\begin{itemize}
\item Case 1. The run is still correct after the renaming because the fragments of $K$ were skipping, from which it follows that the modification of the names does not affect their behavior. The behavior of the other fragments is still the same because we have chosen fresh names. In particular if, at round~$i$, a fragment $F\notin K$ chooses an edge to a fragment $F'\in K$, then this action is still valid as $F$ and $F'$ still have different names.
P2 and the outcome of the run are still correct as we have not changed the fact that the nodes are skipping, and we have not modified the labeling.
Finally, P1 holds because we have considered the first round with a cluster containing more than one fragment, so the predecessors of the fragments of $K$ had different names at the previous rounds. 

\item Case 2. The same line of reasoning as in Case~1 holds. That is,  the behavior is unchanged, the change of name does not affect the correctness of the actions of neither the fragments of $K$, nor the ones outside $K$, and P1 holds because we consider the first round.

\item Case 3. Let us consider first the round that we have added. The fragment $F$ chooses the lightest edge to a fragment that has a different name, because we have chosen different names for all the fragments. Therefore, this round is correct for lazy Bor\r{u}vka algorithm. Now we have to check that the next rounds are also correct. Merging two fragments, we may cause several problems. 
First, the name of this fragment could be incorrectly defined, as the names of the successors of these fragments can be different. However, this cannot be the case because P1 holds in the run before the modification. 
Second, the merged fragment could take two edges at the same round, one taken by the successor of $F$ before the operation, and one taken by the successor of $F'$ before the operation. In fact, this cannot happen, because of P2. Finally the behavior of the other fragments is unchanged as they only consider the names of their adjacent fragments, and these names have not beed modified. Therefore  the run is still correct after the operation. Moreover, we have added a new edge in the labeling, thus the run still describes the labeling at hand. 
Property P1 and P2 still hold for the added round, and they also hold for the next rounds, as we have just merged two fragments of the same cluster, which implies that the names remain unchanged, and if two fragments of a cluster were now choosing an edge at the same round, then this also happened in the run before the operation, which a contradiction. 

\item Case 4. The same reasoning as for the previous case also holds in this case. 
\end{itemize}

\noindent This case analysis completes the proof of Claim~\ref{claim:P1P2preservedafteroperation}. \hfill $\Diamond$

\bigskip 

Thanks to Claim \ref{claim:correct-run}, the labeling $\ell_C'$ is the outcome of a correct run $R$ of lazy Bor\r{u}vka on $C$, and this run satisfies the properties P1 and P2. Therefore we can apply the operation on it. We claim that we can iterate this operation, and eventually get a run $R$ in which there are no clusters with more than just one fragment, after a finite number of iterations. 

To see why, let us first prove that, after an iteration for which we have used Case~3 or Case~4, the number of fragments in the final cluster has decreased by one. Let us consider the two fragments that we have merged during the operation. In the run before the operation, these two fragments had successors that were never merged, because of P1. It follows that these successors were distinct fragments at the end. Now that we have merged them, they form only one fragment at the end. As the behavior of the other fragments is not affected by the change, the number of fragment at the end has therefore indeed decreased by one. 

Let us now prove that for Cases 1 and 2, the sum, over the rounds, of the number of clusters with more than one fragment has strictly decreased. This is because we have scattered such a cluster in both Case~1 and Case~2, without creating new ones. 
Also, the number of fragments in the final cluster remains unchanged. 
Therefore, at every step, either the number of fragment in the final cluster has decreased by one, or it remains unchanged but the sum, over the runs, of the number of clusters with more than one fragment has strictly decreased. The operation can be repeated for a  finite amount of time. 
After all these operations, the run is such that, at every round, no two adjacent fragments have the same name. Therefore, the modified labeling $\ell_C'$ is a spanning tree of $C$.

Overall, we have added exactly one edge every time we have decreased by one the number of fragments in the final cluster. Thus, the number of edges added is equal to the number of fragments in the final cluster in the original run $R$, minus one. This number is at most $k_C$. Indeed, at most one fragment contains no rejecting nodes, since only one fragment can have the node whose ID was used as the root-ID in the certificates, and all the roots of the other fragments reject, with $k_C$ rejecting nodes in total.
Therefore the distance (i.e., the number of modified edges) between the original labeling~$\ell_C'$ and the modified labeling~$\ell'_C$, which is a correct spanning tree of $C$, is at most $k_C$. 
As the same reasoning holds for every connected component,  by defining a ``spanning'' tree of a disconnected graph as the union of   trees spanning each of its connected component, we get that the modified labeling $\ell'=\cup_C \ell'C$ is the spanning tree of $G'$, and it is at distance at most $\sum_C k_C \leq k$ from the original labeling $\ell'=\cup_C \ell_C'$. 

\medskip 

\noindent\textit{Comparison with the original spanning tree.} 
We now compare the modified labeling $\ell'$ with the spanning tree of the original graph $G$.
We claim that the set of edges described by $\ell'$ can be transformed into a spanning tree of $G$ by adding or removing at most $2k$ edges. Indeed, recall that, to go from $G$ to $G'$, we have  removed only the edges that were between two rejecting nodes. Let us call $S$ this set of edges. 
Let us go backwards, and let us remove edges from $G$ to go to $G'$ while keeping track of the spanning tree. 
Among the edges of $S$, at most $k-1$ can be part of the MST of $G$, as otherwise there would be a cycle since there are only $k$ rejecting nodes. Removing the other edges from $G$ does not change the MST. Let $G^1$ be $G$ without these edges, and let us also remove them from $S$. 

Now, let us consider one of the remaining edges in $S$, denoted by $e=\{u,v\}$. Let $G^2$ be the graph $G^1$ without this edge~$e$. 
If removing~$e$ disconnects the graph, then the spanning tree of $G^2$ is the same as the one of $G^1$, without~$e$. 
If removing~$e$ does not disconnect the graph, then we define~$e'$ as the edge of smallest weight in the cut between the nodes that are closer to~$u$ in the tree, and the ones that are closer to~$v$ in the tree. The resulting spanning tree is minimum. To see why, let us check that, for every cycle in the graph, the heaviest edge is not part of the spanning tree. The only cycles that we have to consider are the ones that contain~$e'$. Suppose that the edge~$e'$ is the heaviest of a cycle in $G^2$. This cycle must cross the cut via another edge, and this edge must have a smaller weight, as otherwise $e'$ would not be the heaviest. This contradicts the definition of $e'$, and thus, by adding~$e'$, we obtain a spanning tree of $G^2$. We can repeat this construction until there are no more edges in $S$. At the end, the graph $G^{k}=G'$, and the spanning tree of $G'$ is the modified $\ell'$. We have therefore added or removed at most $2k$ edges.

To conclude, in the first step from $(G,\ell)$ to $(G', \ell')$, we have edited only the labels of the rejecting nodes, thus $k$ labels at most. Then, we got from each labeling $\ell_C'$ to the final labeling $\ell_C'$ by adding at most $k_C$ edges. Thus, in the whole graph, we have modified at most $2k$ labels in total. Finally,  we have added or removed at most $2k$ additional edges, and thus we have modified at most $4k$ labels. Thus, overall, we have edited at most $7k$ labels. It follows that the distance of $(G,\ell)$  to MST is at most linear in the number of rejecting nodes, from which it follows that the proof-labeling scheme is error-sensitive. The sensitivity is 1/7.
This completes the proof of Theorem~\ref{theo:MSTcompactES} \qed

\section{Conclusion and perspectives}
\label{sec:conclusion}

In this paper, we have considered on a  variant of proof-labeling scheme, named \emph{error-sensitive} proof-labeling scheme. We have provided a structural characterization of the distributed languages that can be verified using such a scheme. This characterization highlights the fact that some basic network properties do \emph{not} have error-sensitive proof-labeling schemes, which is in contrast to the fact that every network property has a proof-labeling scheme. However,  important network properties, like cycle-freeness, leader, spanning tree, MST, etc., do admit error-sensitive proof-labeling schemes. Moreover, these schemes can be designed with the same certificate size as the one for the classic proof-labeling schemes for these properties.

Our study of error-sensitive proof-labeling schemes raises intriguing questions. 
In this paper, we have considered two scenarios  only for a given language: either the language does not admit error-sensitive proof-labeling schemes, or it admits an error-sensitive proof-labeling scheme with the same certificate size as the optimal proof-labeling schemes. 

\begin{open}
Is it the case that, for every language, either the language is not error-sensitive, or it is error-sensitive at no cost (that is, the optimal certificate size is the same with and without the error-sensitivity requirement)? 
\end{open}

Another interesting topic is the sensitivity parameter. In this paper, for a given language, we have been interested in the existence of $\alpha$-sensitive scheme for some constant $\alpha$, but we have not tried to optimize~$\alpha$. It would be interesting to study whether we can optimize $\alpha$ beyond the values derived  in this paper.

\begin{open}
In Theorem~\ref{theo:STlogn} and~\ref{theo:MSTcompactES}, we derive error-sensitive schemes with  optimal certificate size, with sensitivity parameters $1/4$ and $1/7$, for spanning tree and minimum spanning tree, respectively. Are these sensitivity parameters optimal?
\end{open}

The two questions above are actually related. 
Indeed, a larger sensibility parameter implies more constraints on the certification. Therefore, in general, certificate size grows with the required sensitivity. Nevertheless, we do not know which way the certificate size grows. 

\begin{open}
For a given language, is it the case that there is a threshold value for the sensitivity parameter (that could be 0), such that, below this threshold, one gets error-sensitivity at no cost in term of certificate size, while, above this threshold, error-sensitivity is impossible? Instead, is there a smooth trade-off between the sensitivity and the certificate size? 
\end{open}

Another desirable property for a proof-labeling scheme is \emph{proximity-sensitivity}, requiring that every error is detected by a node  located closely to the location of the error. Proximity-sensitivity however appears to be a very demanding notion, stronger than error-sensitivity. Indeed, the former implies the latter whenever the errors are spread out in the network at sufficiently large mutual distances. 

\begin{open}
Under which circumstances it is possible to design proximity-sensitive proof-labeling schemes? Is it possible to provide a  simple characterization of the languages admitting   proximity-sensitive proof-labeling schemes?
\end{open} 

Finally, let us mention error-sensitivity for network properties. 
This paper has been motivated by self-stabilizing algorithms, for which mechanisms used to locally certify the correctness of global states of the system are required. 
More recently, there has been a new direction explored in local certification, which consists in designing proof-labeling schemes for properties of the network itself. 
For example, \cite{FeuilloleyFMRRT21, FeuilloleyFMRRT20, EsperetL21} design compact proof-labeling schemes for planar and bounded-genus graphs. 
In such setting, our notion of error-sensitivity is not relevant, as the distance is about the number of edits at the vertices, whereas for network properties, one should consider edits on the graph itself.
A natural question is:

\begin{open}
Can we generalize error-sensitivity to network properties? If yes, is planarity an error-sensitive property?
\end{open} 

\paragraph*{Acknowledgments.} We thank the reviewers of the conference and journal versions of this paper, for their careful reading and useful suggestions.

\bibliographystyle{plain}
\bibliography{bibliography.bib}{}

\end{document}